\documentclass[a4paper]{aart}
\usepackage[centertags]{amsmath}
\usepackage{amsmath}
\usepackage[dutch,british]{babel}
\usepackage{amsfonts}
\usepackage{amssymb}
\usepackage{amsthm}
\usepackage{import}

\usepackage{newlfont}
\usepackage{color}

 \usepackage{bbm}
\include{amsthm_sc}

\usepackage{stmaryrd}
\usepackage{mathrsfs}
\usepackage{pstricks,pst-node}

\usepackage{fancyhdr}
\usepackage{graphicx}
\usepackage{fancybox}

\usepackage{geometry}
\usepackage{setspace}

\theoremstyle{plain}
  \newtheorem{theorem}{Theorem}[section]

  \newtheorem{lemma}[theorem]{Lemma}
    \newtheorem{assumption}{Assumption}
  
\theoremstyle{definition}

\theoremstyle{remark}

\numberwithin{equation}{section}
\numberwithin{figure}{section}

\newcommand{\deq}{\mathrel{\mathop:}=}
\newcommand{\e}[1]{\mathrm{e}^{#1}}
\newcommand{\R} {\mathbb{R}}
\newcommand{\C} {\mathbb{C}}
\newcommand{\N} {\mathbb{N}}
\newcommand{\Z} {\mathbb{Z}}

\newcommand{\adj}{^*} 
\newcommand{\ad}{\mathrm{ad}}
\newcommand{\bra}[1]{\langle #1 |}

\newcommand{\ket}[1]{| #1 \rangle}

\DeclareMathOperator{\Tr}{Tr}

\DeclareMathOperator{\re}{\mathrm{Re}}
\DeclareMathOperator{\im}{\mathrm{Im}}

\newcommand{\caA}{{\mathcal A}}

\newcommand{\caC}{{\mathcal C}}
\newcommand{\caD}{{\mathcal D}}
\newcommand{\caE}{{\mathcal E}}

\newcommand{\caJ}{{\mathcal J}}

\newcommand{\caL}{{\mathcal L}}
\newcommand{\caM}{{\mathcal M}}

\newcommand{\caO}{{\mathcal O}}

\newcommand{\caR}{{\mathcal R}}

\newcommand{\caT}{{\mathcal T}}
\newcommand{\caU}{{\mathcal U}}
\newcommand{\caV}{{\mathcal V}}
\newcommand{\caW}{{\mathcal W}}

\newcommand{\caZ}{{\mathcal Z}}

\newcommand{\bbC}{{\mathbb C}}

\newcommand{\bbH}{{\mathbb H}}

\newcommand{\bbR}{{\mathbb R}}

\newcommand{\bbT}{{\mathbb T}}

\newcommand{\bbV}{{\mathbb V}}

\newcommand{\bbZ}{{\mathbb Z}}

\newcommand{\opunit}{\text{1}\kern-0.22em\text{l}}

\newcommand{\frA}{{\mathfrak A}}

\newcommand{\frX}{{\mathfrak X}}

\newcommand{\scrB}{{\mathscr B}}

\newcommand{\scrE}{{\mathscr E}}

\newcommand{\scrH}{{\mathscr H}}

\renewcommand{\d}{{\mathrm d}}

\newcommand{\Dom}{\mathrm{Dom}}

\newcommand{\beq}{ \begin{equation} }
\newcommand{\eeq}{ \end{equation} }

\newcommand{\baq}{ \begin{eqnarray} }
\newcommand{\eaq}{ \end{eqnarray} }
\newcommand{\bal}{  \begin{align} }
\newcommand{\eal}{\end{align}    }
\newcommand{\bet}{ \begin{theorem} }
\newcommand{\eet}{ \end{theorem} }

\newcommand{\bosondispersion}{\om}

\newcommand{\str}{ |}

\newcommand{\links}{\mathrm{l}}
\newcommand{\rechts}{\mathrm{r}}

\newcommand{\adjoint}{\ad}
\newcommand{\norm}{ \|}

\newcommand{\lone}{\mathbbm{1}}

 \let\be=\beta  

\let\ve=\varepsilon  \let\ga=\gamma 
\let\ka=\kappa \let\la=\lambda \let\om=\omega 
\let\si=\sigma \let\vs=\varsigma

 \let\Ga=\Gamma \let\La=\Lambda

\newcommand{\tdl}{\mathop{\lim}\hspace{0.15cm}\hspace{-0.22cm}_{\Lambda}}

\newcommand{\ben}{\begin{arabicenumerate}}
\newcommand{\een}{\end{arabicenumerate}}

\newcommand{\sys}{\mathrm{S}}
\newcommand{\res}{\mathrm{R}}

\newcommand{\lat}{ \mathbb{Z}^d }
\newcommand{\tor}{ {\mathbb{T}^d}  }

\newcommand{\dd}{\mathrm{d}}
\newcommand{\ii}{\mathrm{i}}

\newcommand{\refer}{{\beta}}
\newcommand{\referres}{{\res,\beta}}

\newcommand{\Matz}{\mathcal{Z}_{[0,t]}}

\newcommand{\kin}{_{M}}
\newcommand{\kini}{_{\widetilde M}}

\newcommand{\eig}{u}
\newcommand{\field}{\chi}
\newcommand{\clasx}{\xi}
\newcommand{\infi}{{t_{-}(I)}}
\newcommand{\supi}{{t_{+}(I)}}

\begin{document}

\begin{center}
\large{ \bf{Quantum Diffusion with Drift and the Einstein Relation I}} \\
\vspace{20pt} \normalsize

\vspace{10pt} 
{\bf   W.\ De Roeck }\\
\vspace{10pt} 
{\bf   Institute for Theoretical Physics \\
Universit\"at Heidelberg\\ 
 69120 Heidelberg, Germany}

\vspace{20pt}

{\bf   J.\ Fr\"ohlich  }\\
\vspace{10pt} 
{\bf   Institute for Theoretical Physics \\
ETH Z\"urich \\
8093 Z\"urich, Switzerland}
\vspace{20pt} 

{\bf K.\ Schnelli  }

\vspace{10pt} 
{\bf   Department of Mathematics \\
Harvard University \\
Cambridge MA 02138, United States}

\vspace{15pt} \normalsize

\end{center}

\vspace{20pt} \footnotesize{ \noindent {\bf Abstract}:  We study the dynamics of a quantum
particle hopping on a simple cubic lattice and driven by a constant external force.  It
is coupled to an array of identical, independent thermal reservoirs consisting of free, massless Bose fields, 
one at each site of the lattice. When the particle visits a site x of the lattice it can emit or absorb field quanta of
the reservoir at x. Under the assumption that the coupling between the particle and the reservoirs and the driving force are sufficiently small, we establish the following results:  The ergodic average over time of the state of the particle approaches a non-equilibrium steady state (NESS) describing a
non-zero mean drift of the particle. Its motion around the mean drift is diffusive, and the
diffusion constant and the drift velocity are related to one another by the Einstein relation. }

\vspace{5pt} \footnotesize \noindent {\bf KEY WORDS:}  diffusion, kinetic limit,
quantum brownian motion  \vspace{20pt}
\normalsize
\section{Introduction}\label{sec: intro}
In this paper and its companion~\cite{paper2}, we study the quantum dynamics of a tracer particle driven by a
constant external force field, $F$, (e.g., a uniform gravitational field) and
interacting with thermal reservoirs by emitting or absorbing gapless reservoir
modes. 
The purpose of our analysis is to establish properties of the long-time
effective dynamics of the particle and to justify some fluctuation-dissipation
relations, notably the Einstein relations.

Among properties of the effective dynamics of the tracer 
particle coupled to thermal reservoirs at positive temperature, we expect 
the following ones to hold:
A velocity-dependent friction force caused by scattering processes
between the particle and the reservoir modes (emission of Cherenkov radiation)
can be expected to counteract the external force driving the particle, in such a way 
that its mean velocity approaches a finite, non-zero limit, $v(F)$, in the direction of the external force $F$,
as time $t$ tends to infinity. Because of thermal fluctuations in the reservoir(s), 
the true motion of the particle is expected to be
\textit{diffusive} around its mean motion. One might guess that the mobility,
$\frac {\partial}{\partial F}\big|_{F=0} v(F)$, of the particle satisfies the
`Einstein relation'
\begin{equation}\label{Einstein}
\frac {\partial}{\partial F}\Big|_{F=0} v(F) = \beta D(F=0)\,,
\end{equation}
where $\beta$ is the inverse temperature of the gas and $D(F)\equiv D(\beta,F)$
its diffusion constant.  This would be the case if we modeled the
dynamics of the particle by a Langevin equation.  However, for a particle moving in physical space~$\mathbb{R}^{3}$ and interacting with a single thermal reservoir filled with a massless, free Bose field (e.g., a weakly interacting Bose gas), this
guess is almost certainly \textit{false}! 
The reason is that, usually, the friction force caused by interactions of the particle with reservoir modes
decreases, as the velocity of the particle grows. One therefore expects
that a sufficiently large force $F$ will eventually overcome the friction force and cause a
`run-away' acceleration. At positive temperature, thermal fluctuations will, in
the long run, always manage to kick the particle velocity into a region where
`run-away' occurs.  In fact, some part of the intuition described here has recently been
made precise in a (classical) Hamiltonian model describing
the system in a limiting regime where the particle is very massive
and the Bose gas in the reservoir is very dense (mean-field limit); see~\cite{FGang,EFGS}. 
In these papers, which only concern systems at zero-temperature, it is shown that the friction
force, $F_{\text{fric}}(v)$, tends to~$0$, as $|v|\rightarrow\infty$, and (assuming
rotational invariance) $|F_{\text{fric}}(v)| = |F_{\text{fric}}( \str v \str)| $
has a unique maximum, $F_{\mathrm{max}}$, at some finite value of $|v|$. 
These properties imply that, for $\vert F\vert <  F_{\mathrm{max}}$, there are two stationary solutions of the equations of
motion corresponding to two different values,
\begin{equation*}
v_{-}\,, v_{+}\,,  \qquad \text{with} \qquad v_{+} > v_{-}\,,
\end{equation*}
of the speed of the particle.
Particle motion with speed $v_{+}$ is likely to be unstable: One expects to find
`run-away' solutions accelerating to
ever higher speed, for initial conditions close to the stationary solution corresponding to~$v_{+}$.

Thus, at positive temperatures, and for a particle moving in $\mathbb{R}^{3}$ interacting 
with a single thermal reservoir, we do \textit{not} expect to ever observe an
approach of the particle's motion towards a uniform mean motion at a
\textit{finite} constant velocity determined by the external force $F$, as time
$t$ tends to infinity, (with diffusive fluctuations around the mean motion).
This type of motion -- if observed -- is a transient phenomenon that may be
encountered at intermediate times, but will disappear at very large times. If the
reasoning sketched here is correct the status of the 
Einstein relation~\eqref{Einstein} becomes rather questionable. 

It actually seems to us that the effective dynamics of a tracer particle in the
continuum driven by a constant external force and interacting 
with an infinitely extended thermal reservoir
(corresponding, e.g., to an ideal or weakly interacting Bose gas at positive
temperature), as described above, is too complicated to be treated with
mathematical precision, at present. We therefore propose to study an \textit{idealized
model} for which the expected phenomenology is much simpler: We imagine that the
tracer particle moves through a static crystalline background of ions modeled by
a periodic potential. We assume that the lowest energy band in the given
periodic potential is well separated from the higher lying bands. Assuming that
the initial state of the particle is a superposition of states with energies
belonging to the lowest band and following the particle's motion only over a
long, but \textit{finite} interval of time, we may neglect transitions of the
particle to states in 
higher-lying energy bands altogether and use a tight-binding approximation to
describe the dynamics of the particle. This amounts to replacing physical space
$\mathbb{R}^3$ by a lattice, e.g., the simple cubic lattice $\mathbb{Z}^3$, and
to considering a quantum particle hopping on $\mathbb{Z}^3$. The particle is
subject to a linear external potential,
$-F\cdot x, x\in \mathbb{Z}^3$, and interacts with a dispersive thermal
reservoir at positive temperature.  The crucial difference,  compared to the
continuum model described above, is that the speed of a quantum particle hopping
on a lattice is uniformly \textit{bounded}, and there are no `run-away
solutions' to the equations of motion.

Yet, the problem of analyzing the long-time motion of the particle in such a
model is still rather challenging, and although a lot of attention was devoted
to problems of this type, see e.g.\ \cite{chen1992quantum, weiss1988dynamics,
sassetti1989linear},  we do not have any precise mathematical results, yet. 
The origin of the difficulties lies in \textit{memory effects} within the
reservoir: The probability for the tracer particle to re-absorb a reservoir mode
it has emitted at another position in space, some time $\Delta t$ ago, does not
decay to $0$ sufficiently rapidly, as $\Delta t$ tends to infinity, to control
the effective dynamics. The reason is that reservoir auto-correlation functions
do not decay in time very fast, uniformly in space. (Assuming they decay in time
integrably fast, uniformly in space, one may hope to face a problem that can be
solved rigorously. However, a solution would still require a major effort; see
\cite{deroeckfrohlichhighdimension,deroeckkupiainendiffusion} for corresponding
results when the external force $F$ vanishes.) We therefore propose to simplify
the problem yet a little further: We assume that, at each site, $x$, of the
lattice $\mathbb{
Z}^{3}$, there is another thermal reservoir, $\res_{x}$, and that reservoirs at
different sites of the lattice are independent. Moreover, all these reservoirs
are isomorphic to one another and are described by free quantum fields at some
positive temperature $\beta^{-1}$. When the particle visits the site $x$ it only
interacts with the reservoir $\res_{x}$. Thus, for memory effects to occur, the
particle has to return to a site it has visited previously. Such memory effects
tend to decay exponentially fast in $\Delta t$.
Models of this sort, but for a vanishing external force field $F$, have been
introduced and studied in~\cite{deroeckfrohlichpizzo}. They are sufficiently
simple to be analyzed rigorously. A somewhat similar model was
treated in \cite{pilletdebievre}.

Assuming that, at this point, the reader has an idea of what the models are that
will be analyzed in this paper (see Section~\ref{sec: model} for precise
definitions), we proceed to summarizing our main results; (for detailed, precise
statements see Section~\ref{sec: results},
Theorems~\ref{thm: stationary} through~\ref{thm: einstein}). We assume
that the strength of interaction between the particle and the reservoirs is
proportional to a coupling constant $\lambda$ that will be chosen appropriately
and that the external force field is given by $F=\lambda^{2}\field$, where $\field \in
\mathbb{R}^3$ is a  fixed vector. The initial state of the system is given by a
product state, $\rho_{\sys}\otimes \rho_{{\res},\beta}$, where $\rho_{\sys}$ is a
density matrix on the Hilbert space, $\mathscr{H}_{\sys} = l^{2}(\mathbb{Z}^3)$, of
the tracer particle localized around some site of the lattice, while~$\rho_{{\res},\beta}$ is the equilibrium state of the reservoirs at inverse
temperature $\beta$. The time evolution of the 
initial 
state in the Schr\"odinger picture is (formally) given by $e^{-\ii
tH}\rho_{\sys}\otimes \rho_{{\res},\beta}\,e^{\ii tH}$, where $H$ is the Hamiltonian
of the system; see Subsections~\ref{particle},~\ref{thereservoir},~\ref{the
interaction}. The precise definition of the Hamiltonian $H$ involves the
dispersion law, $\epsilon$, of the particle and a form factor, $\phi$, that
appears in the interaction Hamiltonian coupling the particle to the reservoirs.
In Subsection~\ref{assume}, two assumptions, on $\epsilon$ and $\phi$, are
formulated that are sufficient for the results summarized below to hold and that
shall not be described here. We define the effective dynamics,
$\mathcal{Z}_{[0,t]}$, of the particle by
\begin{equation}\label{effdyn}
\rho_{\sys,t} \equiv \mathcal{Z}_{[0,t]}\rho_{\sys}\deq \Tr_{\res}\left[ e^{-\ii
tH}(\rho_{\sys}\otimes \rho_{{\res},\beta}) \,e^{\ii tH} \right]\,,
\end{equation}
where $\Tr_{\res}[\,\cdot\,]$ denotes the partial trace over the reservoir degrees of
freedom.  In the thermodynamic limit, the state $\rho_{\res, \beta}$ can no
longer be represented by a density matrix and the above formula needs to be
re-interpreted, but the left-hand side remains meaningful. 

\textit{All} the results summarized below are only known to hold, provided the
coupling constant $\lambda$ is sufficiently small.

Momentum space of the tracer particle is given by the torus $\mathbb{T}^{3}$. 
Let $f$ be an arbitrary  continuous function on momentum space. 
Our first result says that
\begin{equation*}
\lim_{T\to \infty}\frac{1}{T}\int_{0}^{T} \dd t\, 
\Tr_{\sys}\big[f(\,\cdot\,)\rho_{\sys,t}\big] = \langle f,\zeta^{\field,\lambda}
\rangle _{\mathrm{L}^2(\mathbb{T}^3)}\,,
\end{equation*}
where $\zeta^{\field,\lambda}$ is a probability density on $\mathbb{T}^{3}$
describing a stationary state of the tracer particle corresponding to uniform
motion at constant velocity $v\neq 0$. Here, $\Tr_{\sys}[\,\cdot\,]$ denotes the
trace over the tracer particle Hilbert space~$\scrH_{\sys}$. The interpretation of
this result is that the ergodic average over time of the states of the particle
approaches a so-called `non-equilibrium steady state' (NESS). If $\field=0$ then
$\zeta^{\field, \lambda}$ is actually an equilibrium measure on $\mathbb{T}^{3}$
approximately equal to Maxwell's velocity distribution, $\propto
e^{-\beta\epsilon}$, corresponding to the dispersion law $\epsilon$.

Our second result says that, asymptotically, the average motion of the tracer
particle is uniform, with asymptotic velocity
\begin{equation}\label{umotion}
v(\field)= \lim_{t\to \infty} \frac{1}{t} \Tr_{\sys}\big[X\rho_{\sys,t}\big]\,,
\end{equation}
given by  $v(\field)= \langle \nabla \epsilon, \zeta^{\field,\lambda}\rangle$.
Assuming time-reversal invariance, one sees that $v(\field=0)=0$, and one expects
that $\lim_{\field\rightarrow\infty}v(\field) = 0$. To understand the latter, recall
that the Hamiltonian of a  particle hopping on $\bbZ^3$ and driven by a constant
field $F$, has discrete spectrum, corresponding to Bloch oscillations. The
eigenvectors form the so-called Wannier-Stark ladder, they are localized and
their localization length decreases as $\str F\str \to \infty$.  Viewed from
this angle, it is quite remarkable that one gets transport upon coupling to the
reservoirs. 

The third result concerns the fact that the `true' motion of the particle is
diffusive. Because of thermal fluctuations in the reservoirs, the particle
performs quantum Brownian motion around its uniform mean motion, with a
diffusion tensor given by 
\begin{equation}\label{diffusion}
D^{ij}(\field) = \lim_{T\to\infty}\frac{1}{T^2}\int_{0}^{\infty}\dd t\,
e^{-\frac{t}{T}}
\Tr_{\sys}\big[(X^{i}-v^{i}(\field)t)(X^{j}-v^{j}(\field)t)\rho_{\sys,t}\big]\,.
\end{equation}

In the companion paper~\cite{paper2}, we will establish the `Einstein relation'
\begin{equation*}
\frac {\partial}{\partial\field^{i}}\bigg\vert_{\field=0} v^{j}(\field) =
\lambda^{2}\beta D^{ij}(\field=0)\,.
\end{equation*}

Next, we sketch some of the main steps that go into the proofs of these results.
The key idea is to show that the effective dynamics, $\mathcal{Z}_{[0,t]}$, of
the tracer particle defined in \eqref{effdyn} is well approximated by its 
\emph{kinetic limit}. Rescaling space and time as
$(x,t)=\lambda^{-2}(\clasx,\tau)$, the kinetic limit is approached when
$\lambda\rightarrow 0$, with $(\clasx,\tau)$ (arbitrary, but) fixed. Let us
consider the Wigner distribution, $\nu(\clasx,k)$, with $(\clasx,k)\in
\mathbb{R}^{3}\times\mathbb{T}^{3}$, corresponding to the state $\rho_{\sys}$ of
the particle at time $t=0$. Then, in the kinetic limit, time evolution is given
by a linear Boltzmann equation\small
\begin{align} \label{Boltzmann}
\frac{\partial}{\partial \tau}   \nu_{\tau}(\clasx,k)    =  (\nabla
\varepsilon)(k)\cdot\nabla_\clasx
\nu_{\tau}(\clasx,k)-\field\cdot\nabla_k \nu_{\tau}(\clasx,k) +      \int_{\tor} \dd
k'\big[r(k',k)\nu_{\tau}(\clasx,k')-r(k,k')  \nu_{\tau}(\clasx,k )\big]\,,
\end{align} \normalsize
where $r(k,k')dk'$ is the rate for a jump of the particle from momentum $k$ to
momentum $k'$. The kernel $r(k,k')$ can be expressed in terms of the reservoir
auto-correlation function (see Subsections~\ref{assume} and~\ref{sec: strategy
and discussion bis}) and satisfies `detailed balance'. The second, but last term
on the right-hand side of Equation~\eqref{Boltzmann} is a `gain term', the last
term is a `loss term'.

 It is convenient to consider the Fourier transform in the variable $\clasx$ of
Equation~\eqref{Boltzmann}. We set
\begin{equation*}
\hat{\nu}_{\tau}^{\kappa}(k)\deq\frac{1}{(2\pi)^{d/2}}\int_{\R^d}\,\dd
\clasx\,\mathrm{e}^{-\ii
(\kappa,\,\clasx)}\,\nu_{\tau}(\clasx,\,k)\,,
\end{equation*}
where the `pseudo-momentum' $\kappa\in\R^d$ is the variable dual to
$\clasx\in\R^d$. Then $\hat{\nu}_{\tau}^{\kappa}$ satisfies the equation
\begin{equation*}
\frac{\partial}{\partial \tau}\hat{\nu}_{\tau}^{\kappa} =
M^{\kappa,\field}\hat{\nu}_{\tau}^{\kappa}\,,
\end{equation*}
where the operator $M^{\kappa,\field}$ is given by
\begin{equation}\label{kinlim}
(M^{\kappa,\field}
g)(k)\deq\ii\kappa\cdot(\nabla\varepsilon)(k)g(k)-\field\cdot\nabla_{k}
g(k)+ \int_{\tor}\,\dd k' r(k',\,k)g(k')-\int_{\tor}\,\dd k'r(k,\,k')g(k)\,,
\end{equation}
for $g\in \mathrm{L}^{2}(\mathbb{T}^3)$. To understand our approach it is important to know that, for any force field $\field$ and arbitrary $\tau  \geq 0$, 
\begin{align}\label{eq: preview kinetic limit}
(\mathcal{Z}_{[0,\lambda^{-2}\tau ]})_{\lambda^2\kappa} \qquad
\mathop{\longrightarrow}\limits_{\la \to 0} \qquad \e{ \tau
M^{\kappa,\field}}\,,
\end{align}
strongly on $\mathrm{L}^2(\tor)$; see Theorem~\ref{theoremkineticlimit},
Equation~\eqref{eq:04.13}.
The properties of the dynamics generated by the operator~$M^{\kappa,\field}$ are
studied in detail in Section~\ref{sec: kinetic}. It satisfies all the results
summarized above; (approach to a NESS describing uniform motion at a finite
velocity, diffusion around the average motion, Einstein relation, and vanishing
of $v(\field)$, as $\field\to \infty$).

Our goal is then to show that the true dynamics of the tracer particle, as
described by the effective time evolution $\mathcal{Z}_{[0,\lambda^{-2}\tau]}$,
on the fiber corresponding to pseudo-momentum $\lambda^{2}\kappa$ and with
small, but non-zero $\la$, is well approximated by the dynamics in the kinetic
limit, as given by the propagator $\e{ \tau M^{\kappa,\field}}$. Indeed, we will
analyze the properties of $\mathcal{Z}_{[0,t]}$ on the fiber corresponding to
pseudo-momentum $\lambda^{2}\kappa$, for arbitrarily large times~$t$, by viewing
it as a perturbation of the propagator $\e{\lambda^{2}tM^{\kappa,\field}}$. It will
turn out to be convenient to develop the perturbation theory for the Laplace
transforms of $\mathcal{Z}_{[0,t]}$ and $\e{\lambda^{2}tM^{\kappa,\field}}$. The
Laplace transform of the latter is the resolvent of the operator
$\lambda^{2}M^{\kappa,\field}$, while the Laplace transform
of~$\mathcal{Z}_{[0,t]}$ on the fiber with pseudo-momentum $\lambda^{2}\kappa$
is a `pseudo-resolvent' that, in a 
suitable domain of the spectral parameter, can be 
viewed as a small perturbation of the resolvent of $\lambda^{2}M^{\kappa,\field}$.
This analysis is carried out in Section~\ref{sec: analysis of resolvent around
zero}. The formalism in Section~\ref{sec: analysis of resolvent around zero}
relies on material gathered in~\cite{paper2} (standard Dyson
expansion for $\mathcal{Z}_{[0,t]}$ and for time-dependent correlation
functions). The relevant results of~\cite{paper2} are summarized in Section~\ref{section: results from expansions} of the present paper.
 In Section~\ref{sec: equilibrium case}, the system without
external force, i.e., for $\field = 0$, is analyzed in some detail, and, in
Section~\ref{section8.2}, the proofs of the main results are
completed.

Finally, let us also point out an interesting effect that we observe in our
model and that forces us to state our results in the sense of ergodic averages
(note indeed that a more natural expression in ~\eqref{diffusion} would be
$D^{ij}(\field) = \lim_{t\to\infty}\frac{1}{t} 
\Tr_{\sys}\big[(X^{i}-v^{i}(\field)t)(X^{j}-v^{j}(\field)t)\rho_{\sys,t}\big]$.)   
Consider matrix elements $\rho_{\sys,t}(x,x')$ of the reduced density matrix. If
$\str x-x' \str$ is large compared to $\str F \str^{-1}$, then the free
Liouville equation (neglecting the reservoir) is dominated by the driving field
and one can neglect the kinetic term, i.e.,
$\rho_{\sys,t}(x,x') \approx  \e{-\ii F \cdot (x-x')t}\rho_{\sys,0}(x,x')$.  In
that case, the dissipative effect of the reservoir vanishes. Indeed, recall that
in many-body theory dissipation is related to the imaginary part of the
self-energy, which emerges from  virtual transitions. However, if we keep only
the field term in $H_\sys$ (as we just argued), then there are no such
transitions because the coupling to the reservoirs is diagonal in the position
basis.    The conclusion is that these matrix elements $\rho_{\sys,t}(x,x')$ do
not decay with time, as one would naturally expect. In other words, decoherence
is switched off by the field! 
The mathematical expression of this phenomenon is visible in Lemma~\ref{lem:
trivia of delta m}, where we bound the operator $\widetilde M^{\lambda,\kappa,\field}$, which describes
the lowest order (second order) effect of the particle-reservoir coupling. The
point to note is that this operator is not close to the operator $M^{\kappa,\field}$, which
describes the dynamics in lowest order in $\la$. The difference between these two operators
comes from the fact that $\la$ also appears in the force field, since we set $F=\la^2
\field$ and keep $\field$ constant. \\

\textit{Acknowledgements.}
We thank A.\ Kupiainen and A.\ Pizzo for many useful discussions on related
problems. We also thank D.\ Egli, Z.\ Gang and A.\ Knowles for helpful comments.  W.D.R.\ is grateful to the DFG for financial support. 

\section{Definition of the model} \label{sec: model}
In this section, we define our model in a finite volume approximation.
\subsection{Notations and conventions}\label{section2.1.1}
\subsubsection{Banach spaces}
Given a Hilbert space $\scrE$, we use the standard notation 
\begin{equation*} \scrB_p(\scrE)\deq \left\{  A \in \scrB(\scrE)\,:\,
\Tr\left[(A^*A)^{p/2}\right] < \infty  \right\}  ,\qquad   1 \leq p \leq
\infty\,,  \end{equation*}
with $\scrB_\infty(\scrE)\equiv \scrB(\scrE)$ the bounded operators on $\scrE$,
and
\begin{equation*}
\norm A \norm_p \deq \left(\Tr\left[(A^*A)^{p/2}\right]\right)^{1/p}\,, \qquad 
\norm A \norm\deq\norm A \norm_{\infty}\,.
\end{equation*}
 For operators acting on $\scrB_p(\scrE)$, e.g., elements of
$\scrB(\scrB_p(\scrE))$, we often use the calligraphic font: $\caV,\caW$ etc..
An operator $A \in \scrB(\scrE)$ determines  bounded operators
\mbox{$\mathrm{Ad}(A)\,,\adjoint(A)\,,A_{\links}\,,A_{\rechts} $} on $\scrB_p(\scrE)$ by
\begin{align*}
\mathrm{Ad}(A) B\deq ABA\adj\,,\qquad
\adjoint(A) B\deq[A,B]= AB-BA
\end{align*}
and
\begin{align}\label{eq:2.5}
A_{\links}B\deq AB\,,\qquad A_{\rechts}B\deq BA\adj\,,\qquad   B \in
\scrB_p(\scrE)\,.
\end{align}
Note that $(A_1)_{\links}(A_2)_{\rechts}=(A_2)_{\rechts} (A_1)_{\links}$, as
operators on $\scrB_p(\scrE)$, $A_1,A_2\in\scrB(\scrE)$, i.e., the left- and
right multiplications commute. The norm of  operators in $\scrB(\scrB_p(\scrE))$
is defined by \begin{equation*}\label{def: norm on operators}
\norm \caW \norm \deq \sup_{A \in \scrB_p(\scrE)}   \frac{\norm \caW(A)
\norm_p}{\norm A\norm_p}\,.
\end{equation*}
In the following, we usually set $p=1$ or $2$.

\subsubsection{Scalar products}
For vectors $\ka \in \bbC^d$, we let $\mathrm{Re}\, \ka$ denote the vector
$(\mathrm{Re}\, \ka^1, \ldots, \mathrm{Re}\, \ka^d)$, where $\mathrm{Re}$
denotes the real part. Similar notation is used for the imaginary part,
$\mathrm{Im}$. The scalar product on $\bbC^d$ is written as
$(\kappa_1,\kappa_2)$ or  
$\overline{\kappa_1}\cdot\kappa_2$ and the norm as $\str \ka \str \deq
\sqrt{(\ka,\ka)}$. 
The scalar product on  an infinite-dimensional Hilbert space $\scrE$  is written
as $\langle \cdot\,, \cdot \rangle$, or, occasionally, as $\langle \cdot\,,
\cdot \rangle_{\scrE}$. All scalar products are defined to be  linear in the
second argument and anti-linear in the first one.

\subsubsection{Kernels}
 For  $\scrE=\ell^2(\Z^d)$,  we  can  represent
$A\in\scrB_{2}(\scrE)$ by its kernel $A(x,y)$, i.e., $(Af)(x)=\sum_{y} A(x,y)f(y)$, $f\in\scrE$. Similarly, an operator, $\mathcal{A}$, acting on $\scrB_2(\scrE)$ can be represented by
 its kernel $\mathcal{A}(x,y,x',y')$ satisfying
\mbox{$(\mathcal{A}\rho)(x,y)=\sum_{x',y'}\mathcal{A}(x,y,x',y')\rho(x',y')$},
$\rho\in\scrB_2(\scrE)$. 
Occasionally, we use the notation $\ket x$ for $\delta_x  \in \scrE$, defined by \mbox{$\delta_x(x')=\delta_{x,x'}$}, and $\bra x$ for $\langle\,\delta_x\,,\,\cdot\,\rangle$.
 In this
notation $\ket x\bra y$ stands for the rank-one operator $\delta_x\,\langle
\delta_y\,,\,\cdot\,\rangle$. Similarly, for the choice
$\scrE=\mathrm{L}^2(\tor)$, we often use the notation $\ket f$ for
$f\in\mathrm{L}^2(\tor)$ and $\bra g$ for $\langle g,\,\cdot\,\rangle$,
$g\in\mathrm{L}^2(\tor)$. In this `Dirac notation', $\ket f\bra g$ stands for
the rank-one operator $f\langle g,\,\cdot\,\rangle$ on $\mathrm{L}^2(\tor)$.

\subsection{The particle}\label{particle}
Consider the hypercube $ \Lambda=\Lambda_{L}=\bbZ^d \cap [-L/2,L/2]^d$, for some $L\in 2\N$. The particle Hilbert space is chosen as
$\mathscr{H}_{\sys}=\ell^2(\Lambda)$ where the subscript $S$ refers to `system'.
  
To describe the hopping term (kinetic energy), we choose a real function $\ve:
\bbT^d \to \bbR$ and we consider the self-adjoint operator $T\equiv T^{\La} $ on
$\ell^2(\La)$ with symmetric kernel\footnote{{ Later, we will consider $T^{\Lambda}$ as an operator on $\ell^2(\Z^d)$ by the natural embedding of $\ell^2(\La)$ into $\ell^2(\Z^d)$. As such, it has the kernel 
\begin{align}
T^{\Lambda}(x,x')=\begin{cases}
                   \hat\epsilon(x-x')\,,\quad &\textrm{if}\quad x,x'\in \Lambda\\
0 &\textrm{else }
                  \end{cases}\,,
\end{align}
 i.e., we impose Dirichlet boundary conditions.}}
\begin{equation*}
T(x,x') = \hat \ve(x-x')\,,
\end{equation*}
with $\hat \ve$ the Fourier transform of $\ve$. Since we will assume $\ve$ to
be analytic, the hopping is short range.

A natural choice for the dispersion law is $\varepsilon(k) = \sum_j 2 (1-\cos
k^j) $, corresponding to $T=-\Delta$, with $\Delta$ the lattice
Laplacian on $\ell^2(\Lambda)$ with Dirichlet boundary conditions. This choice satisfies all our assumptions, to be stated in 
Section~\ref{assume}.

We define the particle Hamiltonian as
\begin{equation*}
H_{\sys}\deq T-F \cdot X\,,
\end{equation*}
where $F \in \R^d$ is an external force field, e.g., an electric field, and $X\equiv X^{\Lambda}$ 
denotes the position operator on $\scrH_{\sys}$, defined by $Xf(x)=x f(x)$. In what
follows we will write $F= \la^2 \field$, with $\field$ a rescaled field, (a
notation to be motivated later).\\

\subsection{The reservoir}\label{thereservoir}
\subsubsection{Dynamics}
For each $x\in\Z^d$, we define a reservoir Hilbert space at site $x$ by
\begin{align*}
\mathscr{H}_{\res_x}\deq\Gamma_s( \mathrm{L}^2(\bar\Lambda))\,,
\end{align*}
where  $\bar\Lambda=\bar\Lambda_{L}=\R^d\cap[-L/ 2,L/2]^d$ and $\Gamma_s(\caE)$ is the
symmetric (bosonic) Fock space over the Hilbert space~$\caE$. We assume that the
reader is familiar with basic concepts of second quantization, such as Fock
space and creation/annihilation operators; (we refer to,
e.g.,~\cite{derezinski1} for definitions and background).
The total reservoir Hilbert space is defined by
\begin{align*}
\mathscr{H}_{\res}\deq\bigotimes_{x\in\Lambda}\mathscr{H}_{\res_x}\,.
\end{align*}
Note that for all $x$, the spaces $\mathscr{H}_{\res_x}$ are isomorphic to each other. We remark that there is no compelling reason to restrict the one-site reservoirs to
the same region, $[-L/2,L/2]^d$, as the particle system, but this simplifies our
notation. The reservoir Hamiltonian is defined as
\begin{align}\label{def: res ham}
H_{\res}\deq\sum_{x\in\Lambda} \, \, \sum_{q\in\bar\Lambda\adj} \bosondispersion(q)
 a_{x,q}\adj a_{x,q}\,,
\end{align}
where $\bar\Lambda\adj=\frac{2\pi}{L}\Z^d$ is the set of quasi-momenta for the
reservoir at site $x$, and the operators $a^{\#}_{x,q}$ are the canonical
creation/annihilation operators satisfying the commutation relations
\begin{align*}
{ [a_{x,q},a\adj_{x',q'}]=\delta_{x,x'}\delta_{q,q'}\,, \quad\quad[a_{x,q},a_{x',q'}]=[a\adj_{x,q},a\adj_{x',q'}]=0\,,}
\end{align*}
and we choose the dispersion law $\bosondispersion(q)=  \str q
\str + \delta_{q,0}$. Note that this dispersion law corresponds to photons or phonons, except for $q=0$,  where we have modified this dispersion
law at $q=0$ by adding an infrared
regularization that does not affect any of our results; e.g., if we replace
$\delta_{0,q}$ by $K \delta_{0,q}$, with $K >0$, then all  infinite-volume
objects studied in this paper are independent  of $K$.

\subsubsection{Equilibrium state} \label{sec: eq state}
Next, we introduce the {\it Gibbs state} of the reservoir at inverse temperature
$\beta$, $0<\beta<\infty$. It is given by the density matrix 
\begin{equation}\label{eq:2.14}
\rho_{\res, \refer} \deq\frac{1}{Z_{\referres}}\e{-\beta
H_{\res}}\,,\quad\textrm{ where }{Z_{\referres}}=\Tr_{\res}[\e{-\beta
H_{\res}}]\,,
\end{equation}
where $\Tr_{\res}[\,\cdot\,]$ denotes the trace over $\mathscr{H}_{\res}$.

An alternative way to describe this density matrix is to specify the expectation
values of arbitrary observables, which we denote by $ \langle O
\rangle_{\rho_\referres}\deq \Tr_{\res} \left[O\rho_{\res,\refer}\right]$. 
 For $\varphi\in\ell^2(\bar\Lambda^*)$, we write $a_{x}(\varphi)= \sum_{q
\in\bar\Lambda^* }\varphi(q)a_{x,q}$, and we choose observables, $O$, to be
polynomials in the operators $a_{x}(\varphi)$.  One then finds that, for any
$x,x'$ and $\varphi,\varphi'\in\ell^2(\bar\Lambda^*)$:
\begin{itemize}
\item[$i.$] Gauge-invariance:
\begin{equation}\label{eq:gaugeinvariance}
\langle a_x\adj(\varphi)  \rangle_{\rho_\referres} =\langle a_x(\varphi) 
\rangle_{\rho_\referres}=0\,;
\end{equation}
\item[$ii.$]  Two-point correlations: Let
$\varrho_{\beta}\deq(\e{\beta\bosondispersion}-1)^{-1}$, with the one-particle dispersion law
$\bosondispersion(q)=\str q \str+\delta_{q,0}$, be the Bose-Einstein density (operator). Then
\end{itemize}
\begin{align*}
\left( \begin{array}{cc} \langle a^*_{x}(\varphi)
a_{x'}(\varphi')  \rangle_{\rho_\referres}
& \langle a^*_{x}(\varphi) a^*_{x'}(\varphi')  \rangle_{\rho_\referres}  
\\[1mm]
\langle a_{x}(\varphi)a_{x'}(\varphi') \rangle_{\rho_\referres}&\langle
a_{x}(\varphi) a^*_{x'}(\varphi')
 \rangle_{\rho_\referres}
\end{array} \right)  =   \delta_{x,x'} \left(\begin{array}{cc}
 \langle \varphi' ,\varrho_{\beta} \varphi \rangle & 0    \\[1mm]
0 & \langle \varphi , (1+ \varrho_{\beta}) \varphi') \rangle
\end{array}\right)\,;
\end{align*}
\begin{itemize}
\item[$iii.$] Wick's theorem:
 \begin{align}
\langle a^{\#}_{x_{2n}}(\varphi_{2n})  \ldots a^{\#}_{x_{1}}(\varphi_{1}) 
\rangle_{\rho_\referres}  & =  \sum_{\pi\in\mathrm{Pair}(n)} \prod_{(r,s) \in
\pi} \langle a^{\#}_{x_s}(\varphi_s) a^{\#}_{x_r}(\varphi_r) 
\rangle_{\rho_\referres}   \label{eq: gaussian property1}\,,\\[2mm]
\langle a^{\#}_{x_{2n+1}}(\varphi_{2n+1})  \ldots a^{\#}_{1}(\varphi_{1}) 
\rangle_{\rho_\referres}  &=  0\,,\label{eq: gaussian property2}
\end{align}
where $\mathrm{Pair}(n)$ denotes the set of partitions of $\{1,\ldots,2n \}$ into $n$ pairs and
the product is over these pairs $(r,s)$, with the convention that $r<s$. Here,
$\#$ stands either for $\adj$ or nothing.
\end{itemize}

\subsection{The interaction}\label{the interaction}
We define the Hilbert space of state vectors of the coupled system (particle and
reservoirs) by
\begin{equation*}
\mathscr{H}\deq \mathscr{H}_{{\sys}}\otimes\mathscr{H}_{{\res}}\,.
\end{equation*}
We pick  a smooth `structure factor' $\phi\in\mathrm{L}^2(\bbR^d)$ and we define its
finite volume version $\phi^{\La} \in \ell^2(\bar\La^*)$ by
\mbox{$\phi^{\La} (q)\deq  (2\pi/L)^{d/2}\phi(q)$}, with the normalization chosen such that $\norm
\phi\norm_{\mathrm{L}^2(\bbR^d)}= \mathop{\lim}\limits_{L \to \infty}
\norm \phi^\La\norm_{\ell^2({\bar\La^*})} $. We will drop the superscript
$\La$. The interaction between the particle and the reservoir at site $x$ is given by
\begin{equation*}
\lone_x  \otimes  \Psi_x(\phi),  \quad \textrm{where} \quad    \Psi_x(\phi)=  
a_x(\phi)+ {a}\adj_x(\phi)\
\end{equation*}
is the field operator, and $\lone_x=\ket x\bra x$ denotes the projection onto the
lattice site $x$. The interaction Hamiltonian is taken to be
\begin{equation*}
H_{\mathrm{SR}}\deq \sum_{x \in \Lambda}   \lone_x  \otimes \Psi_x(\phi) \quad  
\textrm{on} \quad 
\mathscr{H}_{\sys} \otimes  \mathscr{H}_{\res}\,.
\end{equation*}

The total Hamiltonian of the interacting system on $\scrH$ is then given by
\begin{align}\label{eq:1.13}
H\deq T\otimes\lone-\lambda^2\field\cdot X\otimes\lone+\lone\otimes H_{\res}+\lambda H_{\mathrm{SR}}\,,
\end{align}
where $\lambda\in\R$ is a coupling constant. The interaction term $H_{\sys\res}$
is relatively  bounded w.r.t.\ $H_{\sys}+H_{\res}$ with arbitrarily small
relative bound.  It follows that $H$ is essentially selfadjoint on the domain 
$\scrH_\sys \otimes \Dom(H_\res)$, (where $\Dom(H_\res)$ denotes the domain of
$H_{\res}$).

\subsection{Effective dynamics} \label{effective dynamics}
The time-evolution in the Schr\"odinger picture is given by
\begin{equation*}
\rho_t=\e{-\ii t H}\rho\,\e{\ii t H}\,,\quad\quad\rho\in\mathscr{B}_1(\mathscr{H})\,.
\end{equation*}
We will usually choose an initial state $\rho$ of the form $\rho= \rho_\sys
\otimes \rho_{\res,\refer}$, with $\rho_{\res,\refer}$ as defined above.
Of course, $\rho_t$, with $t>0$, will in general not be a simple tensor product, but we
can always take the partial trace, $\Tr_{\res}[\,\cdot\,]$, over~$\scrH_\res$ to
obtain the `reduced density matrix' $\rho_{\sys,t} $ of the system;
\begin{equation*}
\rho_{\sys,t} =  \Tr_\res \left[ \e{-\ii t H}(\rho_\sys \otimes
\rho_{\res,\refer}) \e{\ii t H} \right]= : \caZ_{[0,t]}\rho_\sys\,,
\end{equation*}
and we call $\caZ_{[0,t]}:
\scrB_1(\mathscr{H}_{\sys})\rightarrow\scrB_1(\mathscr{H}_{\sys}):  \rho_{\sys}
\mapsto \rho_{\sys,t} $ the {\it reduced} or {\it effective} dynamics. 
It  is a trace-preserving and completely positive map.

In the present paper, we will mainly consider observables of the form $O \otimes
\lone$ with $O \in \scrB(\scrH_\sys)$, in which case we can also write
\begin{align}\label{expectationvalue}
\langle O(t)\rangle_{\rho_{\sys}\otimes \rho_{\res,\refer}}\deq\Tr[O(t)\rho_{\sys}\otimes\rho_{\res,\refer}]=\Tr_{\sys}[O
\rho_{\sys,t}]\,,
\end{align}
where the trace $\Tr[\,\cdot\,]$ is over the Hilbert space $\scrH$, the trace  $\Tr_{\sys}[\,\cdot\,]$ is over the particle Hilbert space
$\mathscr{H}_{\sys}$ and~$O(t)$ is the Heisenberg picture time evolution of the observable~$O\otimes\lone$, i.e.,
\begin{align}\label{Heisenberg}
 O(t)\deq\e{\ii tH}(O\otimes\lone)\,\e{-\ii t H}\,.
\end{align}
Note that $O(t)$ is, in general, \textit{not}
of the product form $O' \otimes\lone$, for some $O'$. 

\subsection{Time-reversal} 
We define an anti-linear time-reversal operator $\Theta= \Theta_\sys \otimes
\Theta_\res$, where $\Theta_\sys$ is given by 
\begin{equation*}
\Theta_\sys f(x)  = \overline{f(x)}\,, \qquad  f \in \ell^2(\La)\,,
\end{equation*}
and $\Theta_\res$ by $\Theta_\res\deq \Ga_s(\theta_\res)$, with the one-particle
operator $\theta_\res$ given by
\begin{equation*}
\theta_\res \varphi_{x}(q)  = \overline{\varphi_{x}(-q)}\,, \qquad  \varphi_{x}
\in  \ell^2(\bar \La^*)\,, \qquad x \in \Lambda\,.
\end{equation*}
If the dispersion law $\varepsilon$ of the particle and the form factor $\phi$
are invariant under time-reversal, i.e., $\varepsilon(k)=\varepsilon(-k)$,
$\phi(q)=\overline{\phi(-q)}$ (as will be assumed) then we have that
\begin{equation*}
 \Theta H^{\field=0} \Theta=H^{\field=0}\,,
\end{equation*}
expressing time-reversal invariance of the model.

\section{Assumptions and Results}\label{sec: results}

\subsection{Assumptions}\label{assume}
The model introduced in the last section is parametrized by two functions: the dispersion law $\varepsilon: \bbT^d\to\bbR$, and the form factor
$\phi:\bbR^d\to\bbC$. Here, we formulate our assumptions  on these two
functions.
The  (multi-) strip $\mathbb{V}_{\delta}$ is defined by

\begin{equation}\label{definition of the multi strip}
\bbV_{\delta}\deq\{ z \in (\bbT+\ii \bbT)^d\,:\, |\im z \str \leq \delta \}\,.
\end{equation}

 \renewcommand{\theassumption}{\Alph{assumption}}
 
\begin{assumption} \label{ass: analytic dispersion}\emph{ [Particle
dispersion]} 
The function $\varepsilon$ extends to an analytic function in a region
containing a strip $\bbV_{\delta}, \delta >0$. In particular, the norm
\begin{equation*}
\norm \ve \norm_{\infty,\delta} \deq\sup_{p \in \bbV_{\delta}} \str \ve(p) \str 
\end{equation*}
is finite, for some $\delta>0$. 
Furthermore, there  does not exist any  $v \in\R^d$ such that the function
\[\tor\ni k \mapsto
 ( v, \nabla\varepsilon(k))\]  vanishes identically.
\end{assumption}
This assumption allows us to estimate the free particle propagator
$\e{-\ii t H_{\sys}}$ on the particle Hilbert space \mbox{$\scrH_{\sys}=\ell^2(\La_L)$}
as
follows: 
\begin{align}\label{eq: propagation bound}
\big|\big(\e{-\ii t H_{\sys}}\big)(x,x')\big|\le C\e{-\nu|x-x'|}\e{ t  \norm \im \ve
\norm_{\infty,\nu} }\,.
\end{align}
For $L=\infty$, the bound~\eqref{eq: propagation bound} is
the {\it Combes-Thomas} bound; for finite $L$, it can be established in an
analogous way. If we replace $\Z^d$ by $\R^d$,  any physically acceptable dispersion law
$\varepsilon$ is unbounded, and there is no exponential decay in $|x-x'|$.  This is the main reason why
the system studied in this paper is defined on a lattice.

The next assumption deals with the `time-dependent' correlation function defined (in finite-volume) as
 \begin{equation}\label{eq: clear presentation finite volume correlationfunction one}
\hat \psi^{\La}(t)\deq   \sum_{q \in \bar\La^*}  
\str \phi^{\La}(q)\str^2 \left(  \frac{\e{-\ii t \bosondispersion(q)  }}{\e{\beta\bosondispersion(q)}-1}+\frac{\e{\ii t \bosondispersion(q)}  }{1-\e{-\beta\bosondispersion(q)}}  \right),
\end{equation}
and in the thermodynamic limit as
\begin{equation} \label{eq: clear presentation of correlation function}
\hat\psi(t)\deq
 \int \d q \, \str \phi(q)\str^2  \left(  \frac{\e{-\ii t \str q
\str}}{\e{\beta\str q \str}-1}+\frac{\e{\ii t \str q \str}  }{1-\e{-\beta\str q
\str}}  \right).\
\end{equation}
Since the correlation function $\hat \psi$ is determined by the form factor $\phi$, the following assumption is in fact a constraint on the choice of $\phi$. Define the strip $\bbH_{\beta}$ by
\begin{align}
 \bbH_{\be}\deq \{z \in \bbC\,:\, 0 \leq \im z \leq \be \}\,.
\end{align}
\begin{assumption} \label{ass: exponential decay}\emph{[Decay of reservoir correlation
function]}
The form factor $\phi$ is a spherically symmetric function, i.e., $\phi(q) =:
\phi(\str q \str)$. The correlation functions
$\hat\psi^{\La}(z)$, $\hat \psi(z)$ are uniformly bounded in $ \La$ and $z
\in \bbH_{\be}$, and 
\begin{align*}
\tdl\hat\psi^{\La}(z) = \hat \psi(z) 
\end{align*}
holds uniformly on compacts in $
\bbH_\be$, where $\lim_{\La}$ stands for $\lim_{L \to \infty}$ (recall that $\La\equiv\La_L$).  Furthermore, the number
\begin{equation} \label{eq: gs shift}
\sum_{q \in \bar\La^*}  \bosondispersion(q)^{-1} \str \phi^{\La}(q)\str^2
\end{equation}
is  bounded uniformly in $\La$. Most importantly,  $\hat\psi(z)$ is continuous on $\bbH_\be$ and 
\begin{align*}
\str\hat{\psi}(z)\str \leq C \, \e{-g_\res\str z\str}\,, \qquad z \in \bbH_\be\,.
\end{align*}
\end{assumption}

This assumption mainly states that the
reservoirs exhibit exponential loss of memory. This is a key ingredient for our
analysis.

Often, one also considers the `spectral density'
\begin{equation}\label{definition of psi ohne hut}
\psi (\bosondispersion) = \frac{1}{2 \pi} \int_{-\infty}^{\infty} \d t  \, \hat \psi(t) \,\e{\ii t \bosondispersion}\,.
\end{equation}
It satisfies the so-called `detailed balance' property $\e{\be \bosondispersion}\psi(\bosondispersion)= \psi(-\bosondispersion)$, which
expresses, physically, that the reservoir is in thermal equilibrium at inverse temperature $\beta$.

Assumptions~\ref{ass: analytic dispersion} and~\ref{ass: exponential decay} are
henceforth required and will not be repeated.

\subsection{Thermodynamic limit}\label{sec: thermo}
Up to this point, we have considered a system in a finite volume (cube), $\La$
or $\bar \La$, characterized by its linear size $L$. However, if we wish to
study dissipative effects, we must, of course, pass to the thermodynamic limit,
in order to eliminate finite-volume effects such as Poincar\'e recurrence. This
amounts to taking $\La= \bbZ^d, \bar \La=\bbR^d$ and is accomplished below. 

In this section, we will explicitly put a label $\La$ on all quantities
referring to a system in a finite volume.  As an example, $\scrH_\sys$ now
stands for $\ell^2(\bbZ^d)$, and we write $\scrH^\La_\sys$ for 
$\ell^2(\La)$.   The shorthand $\lim_{\La}$ stands for  the thermodynamic limit,
$\lim_{L \to \infty}$.

\subsubsection{Observables of the system}\label{sec: observables of the system}
We begin by defining some classes of infinite-volume system observables, (i.e., certain
types of bounded operators on $\scrH_{{\sys}}$). We say that an operator $O \in
\scrB(\scrH_\sys)$ is {\it exponentially localized} whenever
\begin{equation*}
\str O(x,x')\str \leq C \e{-\nu (\str x \str+\str x' \str) }, \qquad \textrm{for
some}\, \nu >0\,.
\end{equation*}
An important r\^{o}le is played by the so-called \textit{quasi-diagonal} operators. These are operators $O\in\scrB(\scrH_{\sys})$ with
the property that

\begin{equation*}
\str O(x,x')\str \leq C \e{-\nu (\str x-x' \str) }, \qquad \textrm{for some}\,
\nu >0\,.
\end{equation*}
We denote by $\mathop{\mathfrak{A}}\limits^{\circ} $ the class of quasi-diagonal operators and by $\mathfrak{A} $ its norm-closure.

An observable $O \in \scrB(\scrH_\sys)$ is said to be \textit{translation-invariant} whenever
$\caT_yO=O$, for arbitrary $y \in \bbZ^d$, where $\caT_y O(x,x') \deq O(x+y,x'+y)$.  
Translation-invariant operators on $\scrH_\sys$ form a \textit{commutative}
$C^*$-algebra denoted by $\mathfrak{C}_{\mathrm{ti}}$. We also introduce
the algebras
\begin{equation*}
{\mathop{\mathfrak{A}}\limits^{\circ}}_{\mathrm{ti}}
\deq\mathfrak{C}_{\mathrm{ti}} \cap \mathop{\mathfrak{A}}\limits^{\circ}\,,
\qquad \mathfrak{A}_{\mathrm{ti}} \deq \mathfrak{C}_{\mathrm{ti}} \cap
\mathfrak{A}\,.
\end{equation*}

An operator $O \in  \mathfrak{C}_{\mathrm{ti}} /
{\mathop{\mathfrak{A}}\limits^{\circ}}_{\mathrm{ti}}/ \mathfrak{A}_{\mathrm{ti}}$
can be identified with a multiplication operator, $M_{f}$, on the Hilbert space~$\mathrm{L}^2(\tor)$, i.e., $M_fg=fg$, $g\in\mathrm{L}^2(\tor)$,
with $f:\tor\mapsto \C$ a bounded and measurable/real-analytic/continuous function. Physically, the variable in $\bbT^d$ is the momentum of the particle.

These classes of operators are introduced because certain expansions used in our analysis will apply to quasi-diagonal operators or translation-invariant quasi-diagonal operators, and they can be extended to the closures of these algebras by density.

In analyzing diffusion and in the proof of the Einstein relation we also need to consider certain observables that are unbounded operators: We introduce the
~$^*$-algebra $\mathfrak{X}$ that consists of
polynomials in the components,~$X^{i}$, $i=1,\ldots,d$, of the particle-position operator $X$.

Given an infinite-volume observable $O\in\scrB(\mathscr{H}_{\sys})$,
$\mathscr{H}_{\sys}=\ell^2(\Z^d)$, or $O \in \mathfrak{X}$, we
associate an observable $O^{\Lambda}= \lone_{\La} O \lone_{\La}$ on $\scrH_{\sys}^{\Lambda}=\ell^2(\Lambda)$ with it, where
$ \lone_{\La}$ is the orthogonal projection  $\ell^2(\bbZ^d) \to
\ell^2(\Lambda)$.

\subsubsection{Dynamics}
We choose not to construct directly the time-evolution of infinite-volume observables and
infinite-volume states, although this could be done by using the Araki-Woods representation of the system in the thermodynamic limit. Instead, we will analyze
the infinite-volume dynamics of \textit{`reduced'} states, i.e., of states restricted to
particle observables and correlation (Green) functions of particle-observables by 
constructing these objects as thermodynamic limits of finite-volume expressions.

An infinite-volume density matrix of the particle system
$\rho_{\sys}\in\mathscr{B}_1(\scrH_{\sys})$ is called {\it exponentially localized}
if
\begin{align}
 |\rho_{\sys}(x,x')|\le C\e{-\nu(|x|+|x'|)}\,,\qquad \textrm{for some}\, \nu >0\,.
\end{align}
Given such an infinite-volume density matrix $\rho_{\sys}$, we
associate finite-volume density matrices 
\begin{align}\label{finite volume density matrix}
 \rho_{\sys}^{\Lambda}\deq
\frac{1}{Z_{\rho_{\sys}}^{\Lambda}}\lone_{\La} \rho_{\sys}
\lone_{\La}\in\mathscr{B}_1(\scrH_{\sys}^{\Lambda})\,,\qquad
Z_{\rho_{\sys}}^{\Lambda}\deq\Tr_{\sys}[\lone_{\La} \rho_{\sys}
\lone_{\La}]\,,                                      
\end{align}
 with it. Note that, due to the normalization by
$Z_{\rho_{\sys}}^{\Lambda}$, $\rho_{\sys}^{\Lambda}$ is a density matrix on
$\scrH_{\sys}^{\Lambda}$. 

Recall the definition of the reduced dynamics, $\caZ^{\La}_{[0,t]}$, introduced
in Section~\ref{effective dynamics}, and set
\begin{equation}\label{eqabove}
\caZ_{[0,t]} \rho_{\sys} \deq\tdl  \caZ_{[0,t]}^{\La} \rho_\sys^{\La}  \,.
\end{equation}
The next lemma asserts that the thermodynamic limit  (as $\Lambda$ and $\bar\Lambda$ increase to $\mathbb{Z}^{d}$, $\R^d$, respectively) in~\eqref{eqabove} exists, 
and that the resulting reduced dynamics 
$\caZ_{[0,t]}$ is translation-invariant.
\begin{lemma} \label{lem: thermodynamic dyn}
The limit on the right side of Equation~\eqref{eqabove} exists in
$\scrB_1(\scrH_\sys)$, and this defines the map $\caZ_{[0,t]}: 
\scrB_1(\scrH_\sys)\to \scrB_1(\scrH_\sys)$.  The map $\caZ_{[0,t]}  $ preserves
the trace, i.e.,  $\Tr_{\sys}[ \caZ_{[0,t]} \rho_\sys]= \Tr_{\sys}[ \rho_\sys]$,
positivity and exponential localization of the state of the particle, i.e., if
$\rho_\sys$ has any of these properties, then so does  $\caZ_{[0,t]}
\rho_\sys$. 
Moreover, $\caZ_{[0,t]}$ is translation-invariant; $\caT_{-y}\caZ_{[0,t]}
\caT_y=\caZ_{[0,t]} $ for $y \in \bbZ^d$  with $\caT_y$ as in Subsection~\ref{sec:
observables of the system}.  As a consequence of the above, 
for $O$ in $\mathfrak{A}$ or $\mathfrak{X}$, and for an exponentially localized
state $\rho_\sys$, we can define
\begin{equation*}
 \langle O(t) \rangle_{\rho_\sys \otimes \rho_\referres}\deq \Tr_{\sys} [O \caZ_{[0,t]} \rho_{\sys}]\,. 
\end{equation*}

\end{lemma}

\subsection{Results}\label{sec: small results}
Next, we summarize our \textit{main results}. Throughout this section, it is understood that we consider the infinite-volume system; i.e., $\Lambda=\Z^d$, $\bar\Lambda=\R^d$. 

Our first result describes the approach of
the state of the system to a `non-equilibrium stationary state' (NESS), in the
limit of large times. 

In the theorems below, we use the notation $O(t)$ for
$O^{\field}(t)$, even if $\field \neq 0$. Recall also the multiplication operator $M_f$ on $\mathrm{L}^2(\bbT^d)$ defined in Section \ref{sec: thermo}.

\begin{theorem}\label{thm: stationary}\emph{[Approach to NESS]}
There are constants
$k_\lambda, k_\field, g>0$, such that, for $0 <|\lambda| < k_\lambda$, $|\field|<
k_\field$,  there exists a real-analytic function $\zeta \equiv
\zeta^{\field,\la}$ on $\tor$, satisfying $\zeta\ge0$ and $\int_\tor\dd k\,\zeta(k)=1$, i.e., $\zeta$ is a probability density, such that the following statements hold for any 
exponentially localized density matrix, $\rho_\sys$,  and continuous function $f: \tor \to \bbR$:
\begin{itemize}
\item[$i.$] For $\field\not=0$, 
\begin{equation}\label{results approach to ness 1}
\frac{1}{T}\int_{0}^{T}\,\dd t\,  \langle M_f(t)\rangle_{ \rho_{\sys} \otimes
\rho_\referres}  =     \langle  f, \zeta^{\field,\lambda}  \rangle_{\mathrm{L}^2(\tor)} +
\caO(1/T)\, , \qquad\textrm{ as } T \to \infty\,.
\end{equation}

\item[$ii.$] For $\field\equiv 0$, 
\begin{align}\label{results approach to ness 2}
\langle M_f(t)
\rangle_{ \rho_{\sys} \otimes
\rho_\referres}= \langle f,\zeta^{0,\lambda}\rangle_{\mathrm{L}^2(\tor)} +\caO(\e{-\lambda^2 g t})\,, \qquad\textrm{ as } t \to \infty\,,
\end{align}
 and $\zeta^{0,\lambda}$ satisfies `time
reversal invariance'; $
 \zeta^{0,\la}(k) =    \zeta^{0,\la}(-k) $.
\end{itemize}
\end{theorem}
Our next result asserts that the motion of the particle is \textit{diffusive} around an average uniform motion 
(i.e., a drift at a constant velocity).
 \begin{theorem}\label{thm: diffusion}\emph{[Diffusion]}
 Under  the same assumptions as in Theorem~\ref{thm: stationary}, 
\begin{equation*}
\lim_{t\to \infty}\frac{1}{t}    \langle X(t) \rangle_{\rho_{\sys} \otimes
\rho_\referres}  =v(\field)\,,
\end{equation*}
where $v(\field)$ is the `asymptotic velocity' of the particle and is given by
$v(\field)=\langle\nabla\varepsilon,\zeta^{\field,\la}\rangle$. For
$\field\not=0$, we have $v(\field)\not=0$. The dynamics of the particle is
diffusive,
in the sense that the limits
\begin{equation}\label{eq:diffconst}
D^{ij}(\field) \deq 
\lim_{T\to\infty}\frac{1}{T^2}\int_{0}^{\infty}\,\mathrm{d}t\,\mathrm{e}^{
-\frac{t}{T}}\,  \langle (X^{i}(t)-v^i(\field)t)(X^{j}(t)-v^j(\field)t)
\rangle_{\rho_{\sys} \otimes \rho_\referres}
\end{equation}
exist, where the `diffusion tensor' $D(\field)$ is positive-definite, with
$D(\field)=\caO(\lambda^{-2})$, as $\lambda\to 0$.

 \end{theorem}
Note that the claim about the asymptotic velocity follows formally from
Theorem~\ref{thm: stationary} by defining the velocity operator  as
\begin{equation} \label{def: finite volume velocity}
V^j \deq\ii [ H,X^j]   = \ii [T,X^j]=M_{\nabla^{j} \ve}\,,
\end{equation}
and writing $X(t)=X(0)+\int_0^t \d s V(s)$. Although it is quite easy to make this
reasoning precise, we warn the reader that, at this point, it is
\textit{formal}, because the Heisenberg-picture observables $X(t)$ and $V(t)$ have not been constructed as
operators in the thermodynamic limit. They are formal objects appearing in correlation
functions that are constructed as thermodynamic limits of finite-volume correlation functions.

\subsection{Correlation functions and Einstein relation}
In this section, we present some more results on our model, that are proven in~\cite{paper2}.  We begin with introducing correlation functions. Let $O_1,O_2$ be two observables, i.e., $O_1,O_2$ belong to the algebras $\frA$ or $\frX$; see Section~\ref{sec: thermo}.  For $\Lambda=\Lambda_L$, with $L\in2\N$, we set $O_i^{\Lambda}=\lone_{\Lambda}O_i\lone_{\Lambda}$.  Similarly, given an exponentially localized density matrix $\rho_{\sys}\in\scrB_1(\ell^2(\Z^d))$, its finite-volume version $\rho_{\sys}^{\Lambda}\in\scrB_1(\ell^2(\Lambda))$ is defined in~\eqref{finite volume density matrix}. Let $t_1,t_2\in\R$, then we define the (finite-volume) correlation function as
\begin{align}\label{new correlation function 1}
 \langle O_2^{\Lambda}(t_2)O_1^{\Lambda}(t_1)\rangle_{\rho_{\sys}^{\Lambda}\otimes\rho_{\referres}^{\Lambda}}\deq\Tr[O_2^{\Lambda}(t_2)O_1^{\Lambda}(t_1)(\rho_{\sys}^{\Lambda}\otimes\rho_{\referres}^{\Lambda})]\,,
\end{align}
the trace being over the Hilbert space $\ell^2(\Lambda)\otimes\mathrm{L}^2(\bar\Lambda)$. The infinite-volume correlation function is defined as the limit
\begin{align}\label{tdl of correlation function 1}
 \langle O_2(t_2)O_1(t_1)\rangle_{\rho_{\sys}\otimes\rho_{\referres}}\deq\tdl\Tr[O_2^{\Lambda}(t_2)O_1^{\Lambda}(t_1)(\rho_{\sys}^{\Lambda}\otimes\rho_{\referres}^{\Lambda})]\,,
\end{align}
and we claim that this limit exists for any exponentially localized $\rho_{\sys}$ and any $O_1,O_2$ in $\frA$ or $\frX$. We refer to~\cite{paper2} for a proof of this claim.

Apart from an initial state (density matrix) of the product form $\rho_\sys
\otimes \rho_{\res, \refer}$, we also consider the Gibbs state of the coupled
system when the external force field vanishes, $\field=0$. In finite volume, it is defined by
\begin{equation*}
\rho_\be^{\Lambda} \deq \frac{1}{Z_{\be}^{\Lambda}} \e{-\be H^{\field=0}}, \qquad Z_{\be}^{\Lambda}=  \Tr
\e{-\be H^{\field=0}}, \qquad H^{\field=0} =T\otimes\lone+\lone\otimes H_\res +\la H_{\sys\res}\,,
\end{equation*}
and one easily checks that $\rho_\be^{\Lambda} \in \scrB_1(\ell^2(\Lambda)\otimes\mathrm{L}^2(\bar\Lambda))$. We then define, for $O_1, O_2$ as above,
\begin{align}
 \langle O_1^{\Lambda}(t_1)\rangle_{\rho_\refer^{\Lambda}} &\deq\Tr[ O_1^{\Lambda}(t_1)  \rho_{\refer}^{\Lambda}  ]\label{new expectation}\,,\\
\langle O_2^{\Lambda}(t_2) O^{\Lambda}_1(t_1)\rangle_{\rho_\refer^{\Lambda}} &\deq\Tr[ O_2^{\Lambda}(t_2)O_1^{\Lambda}(t_1)  \rho_{\refer}^{\Lambda}  ]\label{new correlation function 2}\,.
\end{align}
One observes that, for $\field=0$, $\langle O_2^{\Lambda}(t_2+t) O^{\Lambda}_1(t_1+t)\rangle_{\rho_\refer^{\Lambda}}=\langle O_2^{\Lambda}(t_2) O^{\Lambda}_1(t_1)\rangle_{\rho_\refer^{\Lambda}}$, for any $t$, i.e., the correlation functions are time-translation invariant. More generally, one checks that, for $\field=0$,  the correlation function~\eqref{new correlation function 2} satisfies the KMS condition. In particular, we have that
\begin{align}\label{kinder kms}
 \langle O_1(t_1)O_2(t_2)\rangle_{\rho_\be^{\Lambda}}=\langle O_2(t_2)O_1(t_1+\ii\beta)\rangle_{\rho_{\beta}^{\Lambda}}\,,\quad\quad (\field=0)\,.
\end{align}
For $O_1,O_2\in \frA_{\mathrm{ti}}$, the infinite-volume versions of~\eqref{new expectation} and~\eqref{new correlation function 2} are well-defined as the limits
\begin{align}\label{new tld limit equilibrium correlation functions}
\langle O_1(t_1)\rangle_{\rho_\refer}\deq\tdl \Tr[ O_1^{\Lambda}(t_1)  \rho_{\refer}^{\Lambda}  ]\,,\quad\quad\langle O_2(t_2) O_1(t_1)\rangle_{\rho_\refer}\deq\tdl\Tr[ O_2^{\Lambda}(t_2)O_1^{\Lambda}(t_1)  \rho_{\refer}^{\Lambda}  ]\,.
\end{align}
Note that we construct the thermodynamic limit of equilibrium correlation functions only for translation-invariant observables, since, pictorially,
the particle is uniformly distributed in space and hence the expectation values of localized observables vanish. For more details we refer to~\cite{paper2}.

Before we discuss the Einstein relation, let us mention that our model relaxes exponentially fast to equilibrium at vanishing external field; cf., Theorem~3.3.\ in~\cite{paper2}: For $O_1,O_2\in\frA_{\mathrm{ti}}$, $t_1,t_2\in\R_+$, there is $g>0$, such that
\begin{align}
 \langle O_2(t_2)O_1(t_1)\rangle_{\rho_{\sys}\otimes\rho_{\referres}}=\langle O_2(t_2)O_1(t_1)\rangle_{\rho_\be}+\caO(\e{-\lambda^2 gt_1})\,,\quad\quad (t_2>t_1)\,,
\end{align}
  as $t_1\to\infty$, for $\lambda$ sufficiently small and $\field=0$. Moreover, the equilibrium correlation functions~\eqref{new correlation function 2} exhibit the following `exponential cluster property':
\begin{align}
 \langle O_2(t_2)O_1(t_1)\rangle_{\rho_\be}=\langle O_2\rangle_{\rho_\be}\langle O_1\rangle_{\rho_\be}+\caO(\e{-\lambda^2 g (t_2-t_1)})\,,
\end{align}
as $t_2-t_1\to\infty$, for $\lambda$ sufficiently small and $\field=0$. Finally, we mention that the equilibrium correlation function~\eqref{new tld limit equilibrium correlation functions} satisfies the KMS condition on the algebra~$\frA_{\mathrm{ti}}$ of translation-invariant observables; cf., Lemma~3.2 in~\cite{paper2}. Of course, there is nothing special about the restriction to correlation functions with one or two observables and one can prove the statements above for any number of observables.

Our next result states that the equilibrium diffusion matrix $D(\field=0)$
(which is in fact a multiple of the identity matrix) is related to the response
of the particle's motion to the field $\field$. The corresponding identity is
known as the `Einstein relation':
\begin{theorem}\emph{[Einstein relation]}\label{thm: einstein}
 Under  the same assumptions as in Theorem~\ref{thm: stationary}, 
\begin{equation}\label{eq: einstein relation}
\frac{\partial}{\partial
\field^i}\bigg|_{\field=0}v^j(\field)={\lambda^2\beta}D^{ij}(\field=0)\,,
\end{equation}
where $D(\field=0)$ is defined in Equation~\eqref{eq:diffconst} and it equals
\begin{align*}
D^{ij}(\field=0)=\frac{1}{2}\int_{\R}\dd t\,\langle
V^i(t)V^j\rangle_{\rho_\beta}\,.
\end{align*}
\end{theorem}
Note that, by the positivity and isotropy of the diffusion matrix, this theorem
also shows that, for small but non-zero $\field$, $v(\field)$ does not vanish.
The origin of the unfamiliar factor $\la^2$ on the right side of \eqref{eq: einstein
relation} is found in the fact that the driving force field in the Hamiltonian
is $\la^2 \field$, rather than $\field$.

\section{Strategy of proofs and discussion}  \label{sec: strategy and discussion}
Before we are able to present a comprehensible overview of the strategy of the
proofs, we have to introduce some further notions and concepts, such as the
fiber decomposition introduced next.

\subsection{Fiber decomposition}\label{sec: fiber
decomposition}
To start with, we note that $\scrB_{1}(\scrH_\sys)\subset\scrB_2(\scrH_\sys)$,
$\scrH_{\sys}=\ell^2(\Z^d)$. Hence, we may view density matrices on
$\scrH_\sys$ as elements of the space of Hilbert-Schmidt operators,
$\scrB_2(\scrH_\sys)\simeq
\mathrm{L}^2(\bbT^d \times \bbT^d, \d k_{\links}\d k_{\rechts})$. 
We define 
\begin{equation*}
    \widehat O     (k_{\links},k_{\rechts}) \deq\frac{1}{(2\pi)^{d}}
\sum_{x_{\links},x_{\rechts} \in \lat}   
O(x_{\links},x_{\rechts})  \e{- \ii k_{\links} \cdot x_{\links}+\ii k_{\rechts}
\cdot x_{\rechts} }\,, \qquad  O \in    
\scrB_2(\ell^2(\lat))\,.
\end{equation*}
In what follows, we write $O$ for $\widehat O$.
To cope with the translation-invariance of our model, we make the
following
change of variables 
\begin{equation*}
k\deq \frac{k_{\links}+k_{\rechts}}{2}\,, \qquad  p\deq k_{\links}-k_{\rechts}\,,
\end{equation*}
and,   for a.a.\ $p \in \tor$, we obtain a well-defined function $ O_p
\in\mathrm{L}^2(\bbT^d)$ by putting
\begin{equation}\label{eq:2.68}
(O_p)(k) \deq O (k+\frac{p}{2},k-\frac{p}{2})\,.
\end{equation}
This follows from the fact that the Hilbert space  $\scrB_2(\scrH_\sys)
\simeq\mathrm{L}^2(\bbT^d \times \bbT^d, \d k_{\links}\d k_{\rechts})$ can be
represented as a
direct integral
\begin{equation} \label{def: fiber decomposition}
\scrB_2(\scrH_\sys) \simeq \int^\oplus_{ \bbT^d} \d p \,    \scrH_p \,, \qquad  
  O =
 \int^\oplus_{ \bbT^d} \d p \, O_p\,,
\end{equation}
where each `fiber space' $\scrH_p$ can be identified with
$\mathrm{L}^2(\bbT^d)$.
Next, we define, for $\theta=(\theta_{\links},
\theta_{\rechts})\in\C^d \times \C^d$, operators~$\mathcal{J}_{\theta}$ by 
\begin{align} \label{eq:def jkappa}
\mathcal{J}_{\theta}\,O\deq\e{-\ii(\theta_\links,X)}O\,\e{-\ii(\theta_\rechts,
X)}\,,\quad\ O\in\scrB(\mathscr{H}_{\sys})\,.
\end{align}
Note that $\mathcal{J}_{\theta}$ is unbounded  if $\theta$ has an imaginary
part. Also note that a density matrix $\rho_{\sys}\in\scrB_2(\scrH_\sys)$ is
exponentially localized iff $\|\caJ_{\theta}\rho_{\sys}\|_2<\infty$, for
$\theta=(\theta_\links,\theta_\rechts)$ in some complex neighborhood of $(0,0)$.

The following lemma captures some identities used later on. Recall the definition of the strip $\bbV_{\delta}$ in~\eqref{definition of the multi strip}.
\begin{lemma} \label{lemma: fibers}
Let $O \in \scrB_1(\scrH_\sys)$, then 
\begin{equation} \label{eq: trace as integral}
\Tr_{\sys} [O \,\e{\ii p\cdot X }]   =  \langle 1,  O_p   
\rangle_{\mathrm{L}^2(\tor)}=  \int_\tor \d k\, 
O_p(k)\,, \qquad  {\ p\in\tor}\,. 
\end{equation}
If there is a $\delta>0$ such that  $
 \norm \caJ_{\theta/2} O \norm_2  < \infty $, for  $ \str \theta \str  \leq
\delta$, then $p \mapsto  O_p $ is analytic in the interior of the strip
$\bbV_{\delta}  $.
\end{lemma}
(In the discussion above, for $O \in \scrB_2(\scrH_{\sys})$ the fiber operator $O_p$ was defined for a.a.\ $p$, but in the context of Lemma~\eqref{lemma: fibers}, $O_p$ can be defined for any $p$.) The first statement of the lemma follows from the singular-value decomposition for trace-class operators and standard properties of the Fourier transform. The second
statement of Lemma~\eqref{lemma: fibers} is the Paley-Wiener theorem, i.e., the relation between exponential
decay of functions and analyticity of their Fourier transforms; see~\cite{reedsimon2}.

The fiber decomposition in Equation~\eqref{def: fiber decomposition} is useful
when one deals with operators $\caA$ acting on $\scrB_2(\scrH_\sys)$ that are
translation invariant (TI), i.e., 
$\caT_z \caA \caT_{-z}= \caA$, with $\caT_z$ defined  as in Section~\ref{sec:
thermo}.   An important example of a TI operator $\caA$ is the
reduced time-evolution $\caZ_{[0,t]}$; see Lemma \ref{lem: thermodynamic dyn}. For TI operators $\caA$, we find that $(\caA O)_p$ depends on $O_p$ only, and
hence it makes sense to write
\begin{equation}   \label{def: fiber decomposition caa}
(\caA O)_p= \caA_p O_p\,, \qquad  \caA = \int_{\tor}^\oplus \dd p\,\caA_p\,.
\end{equation}

Similarly to Lemma~\ref{lemma: fibers} above, we find that, if $\caJ_{\theta/2}
\caA \caJ_{-\theta/2}$ is bounded for all
$\theta=(\theta_\links,\theta_\rechts)$, with $|\theta|\le \delta$, then the map
$p \mapsto \caA_p$ is analytic in a strip $\bbV_\delta$. Or, in other words, the
kernel of the operator~$\caA$, satisfies
\begin{align}\label{definition locally acting superop}
 |\caA(x_\links,x_\rechts, x'_\links,x'_\rechts)|\le C\e{-\nu
|x_\links-x'_\links|-\nu|x_\rechts-x'_\rechts|}\,,\qquad\textrm{for
}\nu<|\theta|/2\,,\qquad (x_\links,x_\rechts, x'_\links, x'_\rechts\in \Z^d)\,.
\end{align}
In particular,~\eqref{definition locally acting superop} means that $\caA$
preserves the subspace of exponentially localized operators in $\scrB_2(\scrH_{\sys})$.
We call such an $\caA$ a quasi-diagonal operator on $\scrB_2(\scrH_{\sys})$.

\subsection{Strategy of proofs of main results}\label{sec: strategy and
discussion bis}
\subsubsection{Kinetic theory}
For small values of the coupling  constant $\lambda$, one can, at least
heuristically, understand the model studied in this paper with the help of
semiclassical kinetic theory. The reasoning proceeds as follows:  If $\lambda$
approaches zero one must wait a time of order $\lambda^{-2}$ before one sees an
effect of the particle-reservoir interactions. The effect is that the particle
emits or absorbs a field quantum (i.e., a `photon' or `phonon') of one of the thermal
reservoirs and thus changes its momentum.  Since such emission/absorption
processes are well separated in time, one can assume them to be independent, and
this leads to a description of the particle motion in terms of a stochastic
process. 
 Since the maximal velocity of the particle is bounded (this is an effect of the lattice)
and of order one, despite
the presence of the driving field $\field$,
the particle travels a distance of order~$\lambda^{-2}$ during a time of order
$\lambda^{-2}$. 
This motivates the introduction of the \textit{kinetic scale}: We define `macroscopic'
variables, $(\clasx,\tau)$, by setting $\clasx \deq\lambda^{2}{x}$ and $\tau
\deq\lambda^{2}{t}$, where
 the variables $(x,t)$ are the variables used in the definition of the model, henceforth  
 called `microscopic' variables.
The fact that,  for small enough~$\la$, our model is `well-described' by
kinetic theory can be expressed, impressionistically, as follows:
 \begin{equation*}
 \textrm{Hamiltonian evolution}_{\la} (\la^{-2} \clasx, \la^{-2}\tau)\qquad
\mathop{\longrightarrow}\limits_{\la\to 0}  \qquad  \textrm{Stochastic evolution
}(\clasx, \tau)\,.
 \end{equation*}
 The stochastic evolution appearing on the right side is discussed next. 
\subsubsection{Boltzmann equation}\label{sec:strategy:boltzmannequation}
Consider a classical particle with position $\clasx \in \bbR^d$ and (quasi-)
momentum $k \in \bbT^d$. The momentum $k$ evolves
according to a Poisson process with rate $r(k,k') \d k'$ for a jump from
momentum $k$ to momentum $k'$, where $r(k,k')$ is given by 
\begin{equation}\label{jumprates}
r(k,k')\deq\psi[\varepsilon(k')-\varepsilon(k)]\,,
\end{equation}
and $\psi$ is the spectral density given in Equation~\eqref{definition of psi ohne hut}.
 Between two consecutive jumps, at times $\tau$ and $\tau+\Delta
\tau$, the momentum grows linearly in time $k(\tau+\Delta
\tau)=k(\tau)+\field\Delta \tau$ (where
addition is defined on the torus $\tor$). The change in position is
governed by the
(group-) velocity $\nabla\varepsilon(k)$:
\begin{align*}
\clasx(\tau+\Delta \tau)=\clasx(\tau)+\int_{\tau}^{\tau+\Delta \tau}\dd s\,
\nabla\varepsilon(k(s))\,.
\end{align*}  
From this, a Markov process on $\bbR^d \times \bbT^d$ can be constructed using
standard methods.
We present here the associated Master Equation  describing the time-evolution of
the probability density,
$\nu_\tau(\clasx,k)\ge 0$, on phase space $\bbR^d \times \bbT^d$ (with
normalization $\int \dd \clasx\,\int \dd k\,\nu_\tau(\clasx,k)=1$):
\begin{align} \label{eq: transport}
\frac{\partial}{\partial \tau}   \nu_\tau(\clasx,k)    =  (\nabla
\varepsilon)(k)\cdot\nabla_\clasx
\nu_\tau (\clasx,k)-\field\cdot\nabla_k \nu_\tau(\clasx,k) +      \int_{\tor} \dd
k'\big[r(k',k)\nu_\tau(\clasx,k')-r(k,k')  \nu_\tau(\clasx,k )\big]\,.
\end{align} 
For our purposes, it is convenient to consider the  Fourier transform
\begin{equation*}
\hat{\nu}_{\tau}(\kappa, k)\deq\frac{1}{(2\pi)^{d/2}}\int_{\R^d}\,\dd
\clasx\,\mathrm{e}^{-\ii
(\kappa,\,\clasx)}\,\nu_\tau(\clasx,k)\,,
\end{equation*}
where $\kappa\in\R^d$ is the variable dual to $\clasx\in\R^d$. One verifies that
$\hat{\nu}_{\tau}(\kappa, k)$ satisfies the evolution equation
\begin{equation*}
\frac{\partial}{\partial t}\hat{\nu}_{\tau}=M^{\kappa,\field}\hat{\nu}_{\tau},\
\end{equation*}
where, for smooth functions $g$ on $\tor$,
 \begin{align*}
(M^{\kappa,\field}
g)(k)\deq\ii\kappa\cdot(\nabla\varepsilon)(k)g(k)-\field\cdot\nabla_{k}
g(k)+ \int_{\tor}\,\dd k' r(k',\,k)g(k')-\int_{\tor}\,\dd k'r(k,\,k')g(k)\,.
\end{align*}
One can easily check that $M^{\kappa, \field}$ generates a strongly continuous semigroup on $ \mathrm{L}^2(\tor)$.  Its significance in understanding dynamical properties of our model stems from the fact that it describes the
evolution $\mathcal{Z}_{[0,\lambda^{-2}\tau]}$ in the fiber indexed by
$\lambda^2\kappa$, in the limit $\lambda\to 0$. 

\begin{theorem}\emph{[Kinetic limit]}\label{theoremkineticlimit}
For any $\field$ and arbitrary $\tau \geq 0$, 
\begin{align}\label{eq:04.13}
(\mathcal{Z}_{[0,\lambda^{-2}\tau]})_{\lambda^2\kappa} \qquad
\mathop{\longrightarrow}\limits_{\la \to 0} \qquad \e{\tau M^{\kappa,\field}}\,,
\end{align}
strongly on $\mathrm{L}^2(\tor)$. 
\end{theorem}
The restriction to  fibers of order $\la^2$ is equivalent to considering a
macroscopic length scale $\sim \la^{-2}$. 
One can also
convince oneself that $(\rho)_{\lambda^2\kappa}\in\mathrm{L}^2(\tor)$ (the space
which the operator on the left side of~\eqref{eq:04.13} acts on) is the
rescaled Wigner transform of $\rho\in\scrB_2(\scrH_{\sys})$, and one may check
that 
the claim~\eqref{eq:04.13} is equivalent to the results in~\cite{erdos}.

\subsubsection{Perturbation around the kinetic limit} 
The strategy we follow to control the long-time behavior when~$\la$ is small but non-zero, is
basically the same as in~\cite{deroeckfrohlichpizzo}: We represent the Laplace
transform of the fibered dynamics 
$(\mathcal{Z}_{[0,\lambda^{-2}\tau}])_{\lambda^2\kappa} $ as a small (in $\la$)
perturbation of the resolvent of the Boltzmann generator $M^{\ka,\field}$. This
is accomplished by appropriately resumming diagrams, and this is the tedious part of our analysis, which is described in \cite{paper2}. The ideas underlying this analysis are elementary, and our technique is actually a time-dependent counterpart of the use of
`translation-analyticity in the spectral form factor' first applied to the study
of `confined'  open quantum systems by~\cite{jaksicpillet1}. Later, these confined open
quantum systems, where the particle does not have translational degrees of
freedom, have been treated in greater generality; see~\cite{bachfrohlichreturn,derezinskijaksicreturn,merklicommutators}. 

A complication not present in~\cite{deroeckfrohlichpizzo}, is that we
need to keep track of the dependence of the poles in the Laplace transform on 
$\la,\ka,\field $. However, the Hamiltonian at $\field \neq 0$ is not relatively
bounded w.r.t.\ the one at $\field=0$.  This means that we need to develop some
version of asymptotic (rather than analytic) perturbation theory, and this is done in Section~\ref{sec:
analysis of resolvent around zero}.

\subsubsection{Einstein
relation}\label{section:greenkubo} 
In this subsection, we derive the Einstein relation in finite volume (now
dropping the superscript $\La$, because all formulae of this subsection refer to
a finite volume).
We define the velocity operator  as
\begin{equation} \label{def: finite volume velocity finite}
V^j \deq\ii [ H,X^j]   = \ii [T,X^j]\,.
\end{equation}
Note that, because of the finite lattice, this operator is not
translation-invariant; however, its thermodynamic limit is.
Using Duhamel's principle we obtain
\begin{align}\label{duhamel}
\frac{\partial}{\partial \field^i}\bigg|_{\field=0}\langle
V^{j,\field}(t)\rangle_{\rho_\be}=&
-\ii\lambda^2\int_0^t\,\dd s\,\langle
[X^{i,0}(t-s),V^{j,0}(t)]\rangle_{\rho_\beta}.
\end{align}
For simplicity, we drop the spatial indices $i,j$ in the following. Note that
the right-hand side of \eqref{duhamel}  is independent of $\field$.   By
stationarity of the state  $\rho_{\beta}$, it can be written as
$-\ii\lambda^2\int_0^t\,\dd s\,\langle [X^0(-s),V]\rangle_{\rho_\beta}$. In the
remainder of this section, we always set $\field =0$ and we drop this symbol
from our notation. Using the KMS condition we find
\begin{align*}
\int_0^t\,\dd s\,\langle [X(-s),V]\rangle_{\rho_\beta}&=\int_0^t\dd s\,\langle
XV(s)\rangle_{\rho_\beta}-\int_0^t\dd s\,\langle
XV(\ii\beta+s)\rangle_{\rho_\beta}\nonumber\\
&=\ii\int_0^{\beta}\dd u\,\langle XV(\ii u)\rangle_{\rho_\beta}-
\ii\int_0^{\beta}\dd u\,\langle XV(\ii u+t)\rangle_{\rho_\beta}\nonumber\\
&=\ii\int_0^{\beta}\dd u\,\langle XV(\ii
u)\rangle_{\rho_\beta}-\ii\int_0^{\beta}\dd u\,\langle X(-t)V(\ii
u)\rangle_{\rho_\beta}\,,
\end{align*}
where, in the second line, we have used that the integral of the function
$z\mapsto\langle XV(z)\rangle_{\beta}$ vanishes along the contour
$0,t,t+\ii\beta,\ii\beta, 0$. The third line follows by time-translation
invariance. Next,  using $X(-t)=\int_0^{-t}\dd s\, V(s)+X(0)$, by~\eqref{def:
finite volume velocity}, we get
\begin{align*}
 \int_0^t\,\dd s\,\langle [X(-s),V]\rangle_{\rho_\beta}  &= \int_{0}^{\beta}\dd
u\,\int_0^t\dd s\,\langle V(s)V(\ii u)\rangle_{\rho_\beta} \\ 
 &=  \int_{0}^{\beta}\dd u\,\int_0^t\dd s\, \langle V(-s)V(\ii \beta-\ii
u)\rangle_{\rho_\beta}   \\
  &=   \frac{1}{2}\int_{0}^{\beta}\dd u\,\int_{-t}^t\dd s\,\langle VV(s+\ii
u)\rangle_{\rho_\beta} \\
  &=     \frac{\beta}{2}\int_{-t}^{t}\,\dd s\,\langle
VV(s)\rangle_{\rho_\beta}+Q(t)\,.
\end{align*}
The second and third equality follow from time-reversal invariance and the KMS
condition. To arrive at the last equality, we have used that the integral of the
map $z\mapsto\langle VV(z)\rangle_{\rho_\beta}$ vanishes along the contour
$-t,t,t+\ii \beta$, $-t+\ii\beta, -t$, and we have introduced the remainder term
\begin{align*}
Q(t)\deq\frac{\ii}{2}\int_{0}^{\beta}\,\dd u\,\int_0^u\,\dd s\,\langle VV(\ii
s+t)\rangle_{\rho_\beta}-\frac{\ii}{2}\int_{0}^{\beta}\,\dd%
u\,\int_0^u\,\dd s\,\langle VV(\ii s-t)\rangle_{\rho_\beta}\,.
\end{align*}
Recalling our starting point \eqref{duhamel}, we conclude that
\begin{equation*}
\frac{\partial}{\partial \field}\bigg|_{\field=0}\langle
V^{\field}(t)\rangle_{\rho_\be}=
-\ii   \frac{ \lambda^2 \beta}{2}\int_{-t}^{t}\,\dd s\,\langle
VV(s)\rangle_{\rho_\beta} +Q(t)\,,
\end{equation*}
where the dynamics used on the right side is taken at $\field=0$. We now claim that, in the thermodynamic limit, $Q(t) \to 0$, as $t \to \infty$; we refer to~\cite{paper2} for a proof. This proves
the Einstein relation, which relates a non-equilibrium response (left-hand side)
to an equilibrium correlation function (right-hand side). 
\section{Kinetic limit: Linear Boltzmann evolution}\label{sec: kinetic}
As announced in the previous section, we have to study the operator  $M\equiv
M^{\ka, \field}$ introduced in Subsection~\ref{sec:strategy:boltzmannequation}, in
order to unravel properties of the long-time dynamics of the particle.
This operator is of the form
\begin{equation}\label{eq:4.6}
M^{\kappa,\field}=\ii \kappa \cdot (\nabla\varepsilon)- \field\cdot\nabla+G+L\,,
\end{equation}
with the \emph{gain-} and \emph{loss} terms given by
\begin{align*}
(Gg)(k) \deq \int_{\tor}\,\dd k' r(k',\,k)g(k')\,, \qquad  (Lg)(k) \deq
-\int_{\tor}\,\dd k'r(k,\,k')g(k)\,.
\end{align*}
The operator $M^{\kappa,\field}$ is closable (as an operator defined on a domain
dense in $\mathrm{L}^2(\tor)$), because it is a bounded perturbation of the
anti-selfadjoint operator $\field\cdot\nabla$, a core being, e.g., 
$\mathrm{C}^{\infty}(\tor)$. Our main aim is to control the spectrum of
$M^{\kappa,\field}$ and to establish good bounds on the resolvent
$(z-M^{\kappa,\field})^{-1}$. In our analysis, we will emphasize the usefulness of
the $C_0$-semigoup generated by $M^{\kappa,\field}$.

\subsection{Concepts from the theory of $C_0$-semigroups}\label{sec: concepts
czero}
First, we recall some
definitions and results from the theory of strongly continuous semigroups
(hereafter $C_0$-semigroups). For a detailed discussion we refer
to~\cite{engelnagelshortcourse,hillephilips}.  For definiteness, we assume from
the onset that the semigroup acts on $\mathrm{L}^2(\tor)$.
 We say that
$f\in \mathrm{L}^2(\tor)$ is $positive$ ($f\ge0$), or \textit{strictly positive}
($f>0$) iff 
 $f(k)\ge0$, or $f(k)>0$, respectively, for almost every $k \in \tor$.  

A $C_0$-semigroup $(T_t)_{t\ge0}$  is called
\begin{itemize}
\item[--] \emph{positivity-preserving} (or \emph{positive})  if $0\le f $ 
implies $0\le T_t
f$, for each $t\ge0$;
\item [--] \emph{positivity-improving}  if $0\le f, f\not\equiv0$, 
implies $0<T_tf$,  for each $t>0$.
\end{itemize}
For a $C_0$-semigroup $(T_t)_{t\ge0}$, the {\it
growth bound}, $\omega_0$, is defined by
\[\omega_0 \deq\inf\{\omega\in\R\,:\,\exists K_{\omega}, \textrm{ with } 1\leq K_{\omega}<\infty, 
\textrm{ such that }\|T_t\|\le K_{\omega}\e{\omega t}\,,\forall t\ge0\,\}.\] 
Here and in the following we use the symbol $\|\cdot\|$ for the norm on
$\mathrm{L}^2(\tor)$ $and$ for the operator norm on $\scrB(\mathrm{L}^2(\tor))$.
 A $C_0$-semigroup $(T_t)_{t\ge 0}$ has a closed generator,
$A$, and we use the standard notation $T_t=\e{tA}$.  The {\it spectral bound},~$s(A)$, of $A$ is defined as
\[s(A)\deq\sup\{\mathrm{Re}\,z\,:\,z\in\sigma(A)\}\,,\] 
where $\si(A)$ is the spectrum of $A$. 
For any $z\in\C$ with
$\mathrm{Re}\,z > \omega_0$, $(z-A)^{-1}=\int_{0}^{\infty}\dd
t\,\e{-zt}\,T_t$ exists, and we infer the bound 
$\|(z-A)^{-1}\|\le \frac{K_{\omega_0}}{\mathrm{Re} z-\omega_0}$ 
and the inequality $s(A)\le\omega_0$.

Let $O$ be a densely defined closed operator whose spectrum is not all of
$\bbC$. 
We say
(following~\cite{davieslinearoperators}, see also~\cite{katoperturbation}) that
$z\in\C$ belongs to the {\it essential spectrum}, $\sigma_{\mathrm{ess}}(O)$, of
$O$, iff
 $z\lone-O$ is not a Fredholm operator.    We call $r_{\mathrm{ess}}(O)=
\sup\{\str z \str: z \in \sigma_{\mathrm{ess}}(O) \}$ the \textit{essential
spectral radius}.
 A $C_0$-semigroup $T_t$ is called {\it quasi-compact} if $r_{\mathrm{ess}}(T_t)
<1$, for  $t >0$.

\subsection{Spectral analysis of $M^{0,\field}$: Preliminaries}
For now, we neglect the advection term $\ii\kappa\cdot(\nabla\epsilon)$, i.e., we put $\kappa=0$, and study
$M^{0,\field}=-\field\cdot\nabla+G+L$. We therefore omit the superscript $\ka$
everywhere in this section, and we simply write $M^\field\equiv M^{0,\field}$. 
Our main results are summarized in Lemma~\ref{lem: kinetic}.

We say that a function~$f$ on~$\tor$ is real-analytic
if it is analytic in a region containing a multistrip $\bbV_{\delta}$, for some
$\delta>0$. Starting from Assumptions~\ref{ass: analytic dispersion} and~\ref{ass:
exponential decay}, it is straightforward to verify that function $r(k,\,k')$, as defined in~\eqref{jumprates}, is real-analytic in $k$ and $k'$, for some $\delta>0$ determined by $g_\res$ and
$\delta_{\varepsilon}$. Moreover, $r(k,\,k')$ is strictly positive almost everywhere for real arguments. Thus the functions $r(\cdot,k'), r(k,\cdot)$ can vanish
only in isolated points. Therefore,
\begin{align}\label{eq:apm}
a_0\deq\inf_{k\in\tor}\int\,\dd k' r(k,\,k')  >0 \,. 
\end{align}
Setting $L(k)\deq- \int \d k' r(k,k') $ (the loss operator $L$ being defined as
multiplication by $L(k)$), we have that $L(k)\le -a_0$, for all $k\in\tor$. The
strict positivity of the rates $r(\cdot,\cdot)$ implies that the gain operator
$G$ is positivity improving. Moreover, by the smoothness of $r(\cdot,\cdot)$,
$G$ is a compact operator.

\begin{lemma}\label{lemma1}
The operator $-\field\cdot\nabla+L$ generates a positivity-preserving
$C_0$-semigroup $(S_t)_{t\ge0}$ on $\mathrm{L}^2(\tor)$ with growth bound
$\omega_0 \leq  -a_0$, and
\begin{align}
\|(z+\field\cdot\nabla-L)^{-1}\|&\le|\mathrm{Re}\,z+a_0|^{-1}\,,\quad\textrm{for
}\quad
\mathrm{Re}\,z>-a_0\,.\  \label{bound3}
\end{align}
The operator $M^{0,\field}$ generates a positivity-improving $C_0$-semigroup,
$(T_t)_{t\ge0}$, on $\mathrm{L}^2(\tor)$. It has the constant function $1$ as a
left-eigenvector corresponding to the eigenvalue $1$, i.e.,\
 $ \langle 1,T_tf\rangle=\langle 1,f\rangle$, for all $f \in
\mathrm{L}^2(\tor)$ and all $t\geq 0$.
\end{lemma}
\begin{proof} Define the $C_0$-semigroup $(S_t)_{t\ge0}$ by
\begin{align}\label{spectral:05}
(S_tf)(k)\deq f(k-{\field}t)\e{\int_0^t\,\dd s\,L(k+\field(s-t))}\,,\quad
t\ge0\,,\quad f\in\mathrm{L}^2(\tor)\,.
\end{align}
It is easy to check that $-\field\cdot\nabla+L$ is the generator of
$(S_t)_{t\ge0}$, and the growth bound of $(S_t)_{t\ge0}$ is smaller than $
-a_0$.
Since $G$ is bounded, the construction of the semigroup   $(T_t)_{t\ge0}$ is
standard, e.g., by using the norm-convergent Dyson series
\begin{align}\label{eq:10.27}
T_t\deq S_t + \sum_{n=1}^{\infty}\, \mathop{\int}\limits_{0\le t_1<\ldots
t_n\le t} \d t_1 \ldots \d t_n  \, S_{t-t_n}G
S_{t_n-t_{n-1}}G\cdots S_{t_1}\,.
\end{align}
Clearly, we have that $\|T_t\|\le\e{t(-a_0+\|G\|)}$. Observe that the semigroup
$(S_t)_{t\ge0}$ defined in Equation~\eqref{spectral:05} has the property that 
$f>0$ implies $S_tf>0$, for any finite $t\ge0$. Together with the fact that  $G$
is
positivity-improving, this implies that $(T_t)_{t\ge0}$ is positivity-improving,
for any $t>0$. 
One easily checks that, for smooth $f$, 
\begin{align*}
\frac{\dd}{\dd t}\langle 1,f_t \rangle=-{\langle1,\field\cdot\nabla
f_t\rangle}+\langle 1,(G+L)f_t\rangle=0\,, \qquad   f_t \deq T_t f\,,
\end{align*}
(note that both terms vanish separately), and $\langle 1,T_t f\rangle=\langle
1,f\rangle$ holds for arbitrary $f\in\mathrm{L}^2(\tor)$, by a limiting argument. 
\end{proof}
From this lemma, we obtain information on the spectrum of the
 generator $M^{\field}$. 
  \begin{lemma}\label{lem: kinetic}\emph{[Spectrum of $M^{\field}$]}
All statements below hold  for any  $\field\in\R^d$:
\begin{itemize}
\item[$i.$] The essential spectrum of $M^{\field}$ is contained in the region
$\{z\in\C\,:\,\mathrm{Re}\,z\le -a_0\}$. Furthermore, the semigroup generated by
$M^\field$ is quasi-compact. 
\item[$ii.$] Let  $ z \notin \sigma(M^\field)$. If $f $ is real-analytic then
$(z-M^\field)^{-1}f$ is real-analytic, too. 
\item[$iii.$]  
The spectrum of $M^{\field}$ in the region $\{z\in\C\,:\,\re z>-a_0 \}$ consists
of isolated eigenvalues of
finite (algebraic) multiplicity.   The associated eigenvectors  are
real-analytic functions.
\item[$iv.$]  There is a constant $m >0$ (independent of $\field$) such that the region $\{z\in\C\,:\,  \str\im z \str \geq m,\, \re
z > -a_0/2 \}$ does not contain any spectrum.
\item[$v.$]
The only eigenvalue $\mu$ of $M^{\field}$ with $\re \mu \geq 0$ is $\mu=0$, and
it is simple. The spectral projection associated with the eigenvalue $\mu=0$,
$P^{0,\field}\equiv P^{\field}$, is of the form $P^\field=\ket{\zeta^\field}\bra
1$,
where $\zeta^\field$ is a strictly positive function, with $\langle
1,\zeta^\field\rangle=1$.  Moreover, $\sup_\field \norm \zeta^\field \norm
<\infty$. 

\end{itemize}
\end{lemma}
\begin{proof}
$i.$ By Lemma~\ref{lemma1}, we know that $s(-\field\cdot\nabla+L) =-a_0$. 
 Since $G$ is compact, {\it Weyl's theorem} on the stability
of the essential spectrum implies
$\sigma_{\mathrm{ess}}(-\field\cdot\nabla+L)=\sigma_{\mathrm{ess}}
(-\field\cdot\nabla+L+G)$.  Moreover, compact perturbations of generators with
strictly negative growth bound generate quasi-compact semigroups; see
\cite{engelnagelshortcourse}. 

$ii.$ By the analyticity of the function $r(k,k')$ in a strip, we have that   
\begin{equation} \label{eq: analyticity of mf}
\norm \e{\ga\cdot\nabla } M^\field \e{-\ga\cdot\nabla } -M^\field \norm \leq
O(\ga)\,, 
\end{equation}
for sufficiently small $\ga \in \bbC^d$.
Hence, by standard perturbation theory, $\e{\ga\cdot\nabla }(z- M^\field)^{-1}
\e{-\ga\cdot\nabla } $ remains bounded for sufficiently small $\ga$ (depending
on $z$), for any $z \notin \sigma (M^\field)$, and hence
$$
\norm \e{\ga\cdot\nabla }(z- M^\field)^{-1} f \norm \leq  C(z, \ga)  \norm
\e{\ga\cdot\nabla }f \norm\,.
$$ 
The claim on analyticity follows then from the Paley-Wiener theorem.

$iii.$ For any $z\in\C$ with $\mathrm{Re}\,z>-a_0$, we write 
\begin{align*}
z-(-\field\cdot\nabla+L+G)=(z-(-\field\cdot\nabla+L))(\lone-\frac{1}{
z+\field\cdot\nabla-L}G)\,.
\end{align*}
It follows that  $z-(-\field\cdot\nabla+L+G)$ is invertible if and only if
$(\lone-\frac{1}{z+\field\cdot\nabla-L}G)$ is. Since $G$ is
compact, the analytic Fredholm theorem implies that
$(\lone-\frac{1}{z+\field\cdot\nabla-L}G)^{-1}$ is a meromorphic function with
only finitely
or countably many isolated poles of finite (algebraic and geometric)
multiplicity, the
residues of which are finite rank operators. It follows that the spectrum of
$M^{\field}$ in the region $\{z\in\C\,:\, \mathrm{Re}\,z>-a_0 \}$ consists of
isolated eigenvalues of finite multiplicity. Let $M^{\field}f=\mu f$,
$f\not\equiv 0$, with $\re\mu > -a_0$. Since
$\mu\not\in\sigma(-\field\cdot\nabla+L)$, we can rewrite this eigenvalue
equation as
\begin{align}\label{eq: def rewriting eigenvector}
 f=(\mu+\field\cdot\nabla-L)^{-1}{Gf}\,.
\end{align}
 Consequently, for sufficiently small $\ga\in\C^d$, 
\begin{align}\label{eq:10.34}
 \norm \e{\ga\cdot\nabla }  f \norm  &\leq \norm (\mu+\field\cdot\nabla- 
\e{\ga\cdot\nabla } L \e{-\ga\cdot\nabla} )^{-1} \norm  \, \norm
\e{\ga\cdot\nabla } G \norm  \,   \norm f \norm  \nonumber\\
 & \leq  \frac{C}{\str \re \mu- a_0 \str+ \caO(\ga) }  \,    \norm f \norm\,.
\end{align}
Indeed, the  bound on the resolvent follows by Neumann series expansion, using~\eqref{eq: analyticity of mf}, with $M^\field$ replaced by~$L$, and the
resolvent bound on $(z+\field\cdot \nabla -L)^{-1}$ from Lemma~\ref{lemma1},
whereas the bound $\norm \e{\ga\cdot\nabla } G \norm <\infty$ follows from the
analyticity of the kernel $r(\cdot, \cdot)$ of $G$.   By the Paley-Wiener
theorem, $f$ is real-analytic.

$iv.$
   Let 
$M^{\field}f=\mu f$, with $\norm f \norm=1$.  Then, on one hand, 
 \begin{equation} \label{eq:eigenvalue with large f}
 \str \im \mu-  \langle f, \ii \field\cdot \nabla   f \rangle  \str \leq   \str
\re \mu \str+ \norm (G+L)f \norm\,.
 \end{equation}
 On the other hand, by the functional calculus,
 \begin{equation*}
 \str \langle f, \ii \field\cdot \nabla   f \rangle \str \leq  \frac{1}{\nu} 
\langle f, \e{  \nu\str \field \cdot \nabla \str }  f \rangle\,,  \qquad
\nu>0\,.
 \end{equation*}
Since $\re \mu \geq -a_0/2$,  the right-hand side of this equation can, for
sufficient small $\nu$, be bounded independently of $\im \mu$. This follows from
statement $iii$ and Equation~\eqref{eq:10.34}.  Combining this $\im
\mu$-independent bound with \eqref{eq:eigenvalue with large f} yields the claim
$iv$. 

$v.$ The claim that there is a unique simple eigenvalue with maximal real part
(and strictly positive eigenvector~$\zeta^\field$) follows from a
Perron-Frobenius-type argument; see, e.g., Chapter 6, Thm 3.5, in
\cite{engelnagelshortcourse}. This theorem uses the quasi-compactness and the
fact that, for sufficiently large real $z$, $(z-M^{0,\field})^{-1}$ is positivity-improving, which in our case
follows from the fact that $T_t$ is positivity-improving.   The claim that this
eigenvalue is zero, follows then immediately from the relation $\langle 1, T_t
\zeta^\field \rangle=\langle 1, \zeta^\field \rangle $ and the spectral mapping
theorem for generators of quasi-compact semigroups; see, e.g., Chapter~5,
Theorem~4.7 in~\cite{engelnagelshortcourse}. Finally, the uniformity in~$\field$
of the bound on $\Vert\zeta^{\field}\Vert$ follows from~\eqref{eq:10.34},
applied to $f= \zeta^{\field}$, since 
\begin{equation*}
\norm \zeta^\field \norm \leq  \norm (-\field\cdot\nabla +L)^{-1} \norm \norm{G
\zeta^\field} \norm \leq (1/a_0)   \sup_{k'}\|r(\,\cdot\,,k')\|\norm
\zeta^\field \norm_1\,,  
\end{equation*}
where $\norm \zeta^\field \norm_1\deq \int \d k \str \zeta^\field\str= \langle
1, \zeta^\field\rangle=1 $, and we have used the explicit form of $G$. 

\end{proof}

\subsection{Refined spectral analysis of $M^\field=M^{0,\field}$}  \label{sec:
refined}
In this subsection, we refine the conclusions of Lemma~\ref{lem: kinetic}.  The
main difficulty we face is that the operator $M^{\field}$ is not analytic (in
any reasonable sense) in the parameter~$\field$, and hence, a priori,
perturbation theory does not apply to the isolated eigenvalue at $0$ and the
corresponding eigenvector. This difficulty is overcome in the next  lemma that
shows that the resolvent  and the spectrum of $M^\field$ can be controlled in terms
of $M^0$, for $\field$ sufficiently small.  Note that this would be obvious if
$M^{\field}$ were analytic in $\field$.   Afterwards, we also comment on the
case of large $\field$. 

\begin{lemma}\label{lem: resolventbound} 
There are constants $C\,,C'$ such that
\begin{align}\label{eq:5.47}
\bigg\|\frac{1}{z-M^{\field}}\bigg\|\le\frac{C}{|\mathrm{Re}\,z+g\kin
+\caO(\field)|}+\frac{C'}{|z|}\,,\quad\textrm{ for } \quad\mathrm{Re}\,z> -g\kin
+\caO(\field)\,,
\end{align}
as $\field\rightarrow 0 $, 
where $g\kin \deq \mathrm{dist} (\ii\R,\si(M^0) \setminus \{ 0\}) $.  
\end{lemma}
In the proof of this lemma we make use of the transformation
$B_{\varepsilon}\deq\e{\frac{\beta}{2}\varepsilon} B
\e{-\frac{\beta}{2}\varepsilon} $, for any operator~$B$ acting on~$\mathrm{L}^2(\tor)$. Here $\varepsilon$ is the dispersion law of the particle. Since $\varepsilon$ is positive and
bounded, it immediately follows that the spectra of $B_{\varepsilon}$ and $B$
coincide. In particular, we will consider
\begin{equation}\label{eq:10.48}
M_{\varepsilon}^{\field}=G_\varepsilon
+L+\frac{\beta}{2}
(\field\cdot\nabla\varepsilon)-\field\cdot\nabla\,.
\end{equation}
This transformation is useful, because the  rates satisfy the identity
\begin{align}\label{detailedbalance}
r(k,k')=r(k',k)\,\e{-\beta(\varepsilon(k')-\varepsilon(k))}\,,
\end{align}
where $\beta$ is the inverse temperature of the reservoirs. This identity is
known as the
\emph{detailed balance} condition.  It is a consequence of the
KMS condition for the reservoirs and the time-reversal symmetry; (recall
the discussion following Assumption~\ref{ass: exponential decay} in
Section~\ref{assume}).
Equation~\eqref{detailedbalance} implies that
${G}_{\varepsilon}$ and
${G}_{\varepsilon}+L+\frac{\beta}{2}(\field\cdot\nabla\varepsilon)$ are
 selfadjoint. But note that $-\field\cdot\nabla$ is anti-selfadjoint. By
$P^\field$ we denote the spectral projection associated to the eigenvalue
$\mu=0$ of $M^\field$, and we set $\bar{P}^\field\deq\lone-P^\field$. For
$\field=0$, we infer from the detailed balance condition that
$P^0=\ket{\zeta^0}\bra 1$, with $\zeta^{0}$ the `Gibbs state'
$\zeta^0(k)=\frac{1}{\langle
1,\e{-\beta\varepsilon}\rangle}\e{-\beta\varepsilon(k)}$. Finally, we note that
$P^0_{\varepsilon}$ and $\bar{P}^0_{\varepsilon}$ are orthogonal projections.

\begin{proof}[Proof of Lemma~\ref{lem: resolventbound}]

We split
\[(z-M^{\field})^{-1}=P^\field(z-M^{\field})^{-1}P^{\field}+\bar
P^\field(z-M^{\field})^{-1}\bar P^{\field}\,,\]
and we remark that the first term is bounded by $\frac{C}{|z|} $, since the
eigenvalue $\mu=0$ is simple and $\norm \zeta^\field \norm$ is bounded uniformly
in $\field$, by Lemma~\ref{lem: kinetic}, $v$. To deal with the second term, 
note that the left-eigenvector at $\mu=0$ does not depend on $\field$, i.e.,
$P^{\field}=\ket{\zeta^{\field}}\bra 1$, which implies
\begin{align}\label{relationP}
P^\field P^{\field'}=P^{\field}\,,\quad\bar P^\field\bar P^{\field'}=\bar
P^{\field'}\,,\quad P^\field\bar P^{\field'}=0\,,
\end{align}
 (however, $\bar P^\field P^{\field'}\not=0$, for $\field\not=\field'$). In
particular, we have that
 \begin{equation*}
 M^{\field}=   \bar P^0 M^{\field} \bar P^0   +\bar P^0 M^{\field}P^0\,,
 \end{equation*}
 and,  by straightforward algebra,
 \begin{equation*}
 (z-M^{\field})^{-1}  =   {z}^{-1}P^0+\bar P^0(z\bar P^0-\bar P^0 M^{\field} 
\bar P^0 )^{-1}\bar P^0+\bar P^0(z \bar P^0-  \bar P^0 M^{\field}  \bar P^0
)^{-1}      \bar P^0 M^{\field}P^0.    
 \end{equation*}
Using \eqref{relationP}, this leads to
\begin{equation*}
\bar P^\field (z-M^{\field})^{-1}\bar P^\field=\bar P^\field (\bar P^0(z- 
M^{\field})  \bar P^0 )^{-1} \bar P^\field\,.
\end{equation*}
We recall the conjugation $B\to B_{\ve}$, introduced above, and write
\begin{align*}
( \bar P^0_\varepsilon(z-M^{\field}_\varepsilon) \bar P^0_\varepsilon)^{-1} = 
(R+\ii I)^{-1} \,,
\end{align*}
where $R$ is defined by $R \deq \bar P^0_\varepsilon ( \re z -(G+L)_\varepsilon
+\frac{\beta}{2}
(\field\cdot\nabla\varepsilon)  )\bar P^0_\varepsilon $ and $I \deq  \bar
P_\varepsilon^0 ( \im z +\ii\field \cdot \nabla  )\bar P_\varepsilon^0$. Note
that~$R$ and~$I$ are selfadjoint operators. By the spectral calculus and the
boundedness of $\field\cdot\nabla\varepsilon$,
\begin{align*}
R\ge \re z-g\kin+\caO(\field)\,.
\end{align*}
Hence, $R>0$ (strictly positive in the sense that $\inf\sigma(R)>0$), for $\re
z>g\kin+\caO(\field)$, and we find that
\begin{align*}
 (R+\ii I)^{-1}= \frac{1}{\sqrt{R}} \left(1 +\ii\frac{1}{\sqrt{R}} 
I\frac{1}{\sqrt{R}} \right)^{-1}   \frac{1}{\sqrt{R}}\,.
\end{align*}
Using the selfadjointness of $I$, the boundedness of $\e{\beta\varepsilon/2}$
and $P^\field$ (uniform in $\field$), we obtain the bound
\begin{equation*}
\left\norm  (z-M^{\field})^{-1}\bar P^\field \right\norm \leq  C \left\str g\kin
+ \caO(\field)+ \re z\right\str^{-1}, \qquad \re z >  -g\kin+ \caO(\field)\,.
\end{equation*} 
\end{proof}

As promised, we now turn to a discussion of the model for large $\field$. For simplicity, we restrict
our attention to the one-dimensional case, $d=1$, and comment on higher
dimensions at the end of this discussion. 

\begin{lemma}\label{lem: perturbation large field}
 Let $d=1$. There are constants $C,C'$ such that,  for sufficiently large
$\str\field \str $,
  \begin{equation} \label{eq: resolvent bound large fields}
\left \norm \frac{1}{z-M^\field} - \frac{1}{z} Q  \right\norm \leq \frac{C}{\str
\field \str}\,,
 \end{equation}
 for  $z \in \bbC$, with $\str\im z\str \leq C' $. 
 Here $Q$ is the orthogonal projection on the space of constant functions
on~$\bbT$. 
\end{lemma}

\begin{proof}
In $d=1$, the spectrum of the self-adjoint operator $\field\cdot X$ is the
lattice $\field \,\bbZ$, and each of the eigenvalues corresponds to a
one-dimensional eigenspace. The eigenspace corresponding to $0$ is the space of
constant functions. Hence, we may write
\begin{equation*}
M^{\field} =     Q M^0 Q   +  \bar Q  (\field\cdot X)   \bar Q   +    W\,,
\qquad  \textrm{with} \qquad  W =     Q M^0  \bar Q  +    \bar Q M^0  Q\,,
\end{equation*}
and $\bar Q = \lone -Q$.   As argued previously,   $Q (G+L) =0$  (the constant function is a left
eigenvector of $M^\field$). Since $\bar Q  (\field\cdot
X)\bar Q $ has spectrum in the region $\str\im z\str \geq \str \field \str$ and
$W$ is bounded, we obtain the claim of the lemma by a straightforward Neumann
series expansion. 
 \end{proof}

In higher dimensions, things are more subtle. If $\field$ is a multiple of some element in
$\bbZ^d$, then the above lemma is  easily generalized by replacing $Q$ with the
(infinite-dimensional) projection corresponding to the kernel of $\field \cdot
X$ and  the term   $\frac{1}{z} Q$  in \eqref{eq: resolvent bound large fields}
by $ Q \frac{1}{z- QM^0 Q } Q $. In this situation, the spectrum of $\field
\cdot X$ is a lattice and the gap between $0$ and the rest of the spectrum
increases proportionally to $\str \field \str$. Since the spectral analysis of
$QM^0 Q $ can be carried out similarly to that of $M^0$, this allows us to
control the eigenvalue at $0$ by perturbation theory, as $\field\to \infty$.  
 However, if the line $\field \bbR$ does not hit any lattice point $\bbZ^d$,
then the spectrum of $\field\cdot X$ covers the whole real line and we need more
subtle considerations to perform the limit $\field\to \infty$.  This is expected to be
manageable; but we do not wish to address this point here.

\subsection{Asymptotics of the semigroup $\e{tM^\field}$}  \label{sec:
asymptotics of semigroup}
In this section, we discuss the large time asymptotics of the semigroup
generated by $M^{\kappa,\field}$. In particular, we show that the diffusion
tensor at vanishing external field is positive-definite.

Our interest is in the asymptotic behavior of solutions of the linear Boltzmann
equation, i.e.,\ of the probability density $\nu_t(x,k)$ satisfying the
evolution Equation~\eqref{eq: transport}.  The distribution $ \hat \nu^0_{t}(k) \deq  \int \d x \,\nu_t(x,k) $ of the
particle's momentum  evolves according to
the semigroup $\e{t M^{0,\field}}$. We recall that this semi-group is
quasi-compact (and positivity improving) and that $\mu=0$ is an isolated
eigenvalue of $M^{\field}$, by Lemma~\ref{lem: kinetic}. Hence, we conclude
(see, e.g.,~\cite{engelnagelshortcourse}) that, for any $\hat \nu^0_{t=0} \in
\mathrm{L}^2(\tor)$, 
\begin{equation*}
\hat \nu^0_{t}  = \zeta^{\field} +  \caO(\e{-g\kin(\field) t })\,, \qquad  t \to
\infty\,,
\end{equation*}
where $g\kin(\field)= \mathrm{dist}(\ii \R,\si(M^{0,\field}) \setminus \{0\})$.
We know that $g\kin(\field)>0$, for all $\field$; but only for small $\field$ we
have established uniformity in~$\field$; see Lemma~\ref{lem: resolventbound}. 
Information on the distribution of the particle's position is obtained from operators on
fibers at non-zero $\ka$, (as explained in Section~\ref{sec: strategy and
discussion}). Hence we have to determine the asymptotic behavior of $\e{t
M^{\ka,\field}}$, for small $\ka$. Recall that
$M^{\ka,\field}=M^{0,\field}+\ii\ka\cdot\nabla{\ve}$. Since
$\ii\ka\cdot\nabla{\ve}$ is a bounded operator, analytic perturbation theory in
$\kappa$ implies that the operator $M^{\ka,\field}$ has an isolated, simple
eigenvalue, $\eig\kin(\ka)$, close to $0$, for small $\kappa\in\C^d$. Moreover,
by the resolvent bound in Lemma~\ref{lem: resolventbound}, 
$g\kin(\ka,\field)\deq \mathrm{dist}(\ii\R,\si(M^{0,\field}) \setminus
\{\eig\kin(\ka)\})$ equals  $ g\kin(\field)+\caO(\ka)$. For $\ka$ small enough,
the semigroup generated by~$M^{\ka,\field}$ is quasi-compact, and we conclude,
by similar reasoning as above, that
\begin{equation*}
\e{ t M^{\ka,\field}} = P\kin^{\ka,\field} \e{t \eig\kin(\ka,\field)} +
\caO(\e{t (\eig\kin(\ka,\field)-g\kin(\ka,\field))})\,,
\end{equation*}
with $P\kin^{\ka,\field}$ a rank-one operator that is a small perturbation of
$\ket{\zeta^{\field}}\bra 1$. 

Following the discussion in Section~\ref{sec: strategy and discussion}, we know
that the asymptotic velocity $v\kin(\field)$ and diffusion constant
$D\kin(\field)$ can be derived from $\eig\kin$, with 
\begin{equation*}
v^i\kin(\field)= \ii\frac{\partial }{\partial \ka^i}\bigg|_{\kappa=0}
\eig\kin(\ka,\field), \qquad   D^{ij}\kin(\field)=-\frac{1}{2} \frac{\partial
}{\partial \ka^{i} \partial \ka^j}\bigg|_{\ka=0} \eig\kin(\ka,\field)\,.
\end{equation*}
These expressions can be computed using the standard Rayleigh-Schr\"odinger
expansion of analytic perturbation theory:
\begin{align*}
 \eig\kin(\ka,\field)=\ii\langle
1,(\kappa,\nabla\epsilon)\zeta^{0,\field}\rangle-\langle
1,(\kappa,\nabla\epsilon)S^{\field}(\kappa,\nabla\epsilon)\zeta^{0,\field}
\rangle+\caO(\kappa^3)\,,\quad |\ka|\rightarrow 0\,,
\end{align*}
where by $S^{\field}$ we denote the `reduced resolvent' of $M^{0,\field}$ at
$z=0$, i.e., $S^{\field}=(\field\cdot\nabla-G-L)^{-1}\bar P^{\field}$, with
$\bar P^{\field}=\lone-P^{\field}$, $P^{\field}=\ket{\zeta^{\field}}\bra {1}$.
Note that
$\overline{\eig\kin(\ka,\field)}=\eig\kin(-\overline{\kappa},\field)$, so that
$v\kin(\field)$ and $D\kin(\field)$ have real entries.

For $\field=0$, the detailed balance condition implies that
$\zeta^{\field=0}(k)\propto \e{-\beta\epsilon(k)}$, hence $v\kin(\field=0)=0$.
To prove that the diffusion constant at vanishing external field is strictly
positive, we use the transformation $B\mapsto B_{\epsilon}$ (as defined in
Section~\ref{sec: refined}) to find that, for any $a\in\R^d$,
\begin{align*}
 (a,D_M(\field=0)a)=-\frac{1}{2}\langle
1,(a,\nabla\epsilon)S^{0}(a,\nabla\epsilon)\zeta^{0}\rangle=-\frac{1}{2\langle
1,\e{-\beta\epsilon}\rangle}\langle
\e{-\beta\epsilon/2},(a,\nabla\epsilon)S^{0}_{\epsilon}(a,\nabla\epsilon)\e{
-\beta\epsilon/2}\rangle\,.
\end{align*}
By the spectral calculus and Lemma~\ref{lem: kinetic}, we know that
$-S^0_{\epsilon}$ is strictly positive. Moreover, by Assumption A,
$(a,\nabla\epsilon)$ does not vanish identically. It follows that
$D\kin(\field=0)$ is a positive-definite matrix.

For $\field\not=0$, one can establish smoothness (but $not$ analyticity) of
$\eig\kin(\ka,\field)$, $v\kin(\field)$ and $D\kin(\field)$ in $\field$, using
asymptotic perturbation expansions in $\field$, for $\field$ small enough. This
method is outlined in Lemma~6.1 of~\cite{paper2}. Relying on these
expansions and using the ideas just presented, it is straightforward to prove
that $(\field,v\kin(\field))\not=0$, for $\field\not=0$, and that
$D\kin(\field)>0$, for small external forces $\field$. Moreover, one can also
confirm the validity of the Einstein relation within the kinetic theory:
$\frac{\partial }{\partial \field}\big|_{\field=0}v\kin(\field)= \beta
D\kin(0)$. 
\\

 Finally, let us turn to the large-$\field$ regime. In dimension $d=1$,
Lemma~\ref{lem: perturbation large field} allows us to apply perturbation theory
in the parameter $1/\str \field \str$, and we derive easily that both
$v(\field)$ and $D\kin(\field)$ vanish as~$1/\str \field \str$, for $\str \field
\str \to \infty$.
The main gap in our knowledge is for moderate $\str \field \str$: We are not
able to prove that  $D\kin(\field) >0$. The fact that this is difficult using
spectral methods should not come as a surprise. In fact, modern approaches to the central
limit theorem often use martingale techniques. 
However, a standard method we are aware of, which one may want to apply to prove
the positivity  of the diffusion constant and a central limit type theorem,
the \emph{graded sector condition} introduced in \cite{komorowski2003sector},
does not appear to be applicable here.

\section{Results from expansions}\label{section: results from expansions}
The aim of this section is to summarize properties of the effective dynamics $\caZ_{[0,t]}$ that can be proven using expansion techniques. We only describe the main ideas and present formal arguments. Mathematically precise arguments are given in~\cite{paper2}, Sections~4 and~5, where elaborate expansion techniques are developed that can also be used to analyze correlation functions.

\subsection{Survey of expansions}
Here we sketch expansion techniques that are used to study the reduced dynamics of the tracer particle. Let
$I\subset\R_+$ be a finite interval. We define the free particle dynamics,
$\caU^{{\Lambda}}_I$, on $\scrB_2(\scrH_{\sys})$, with $\scrH_{\sys}=\ell^2(\Lambda)$, by
\begin{align}
 \caU^{{\Lambda}}_I\deq\e{-\ii |I|\ad(H_{\sys})}\,,\quad\quad
H_{\sys}\equiv H_{\sys}^{\Lambda}=T^{\Lambda}-\lambda^2\field\cdot X^{\Lambda}\,,
\end{align}
and the particle-reservoir interaction, $H_{\mathrm{SR}}(t)\equiv H_{\mathrm{SR}}^{\Lambda}(t)$ in the interaction picture,
which we may write as a sum over spatially localized terms, by
\begin{align}
 H_{\mathrm{SR}}^{}(t)\deq\e{\ii t H_{\res}}\,H_{\mathrm{SR}}\e{-\ii t
H_{\res}}=\sum_{x\in{\Lambda}}\lone_{x}\otimes\e{\ii t
H_{\res}}(a_x(\phi)+a\adj_x(\phi))\,\e{\ii t H_{\res}}\,.
\end{align}
Iterating Duhamel's formula

\begin{align*}
\e{\ii t\ad(H_{\res})}\e{-\ii t
\ad(H)}&=\mathcal{U}_{[0,t]}-\ii\lambda\int_0^t\dd
s\,\mathcal{U}_{[s,t]}\,\ad(H_{\mathrm{SR}}(s))\,\e{\ii s \ad(H_{\res})}\,\e{-\ii
s\ad(H)}\,,
\end{align*}
we find the (Lie-Schwinger-) Dyson series for $\mathcal{Z}_I$:\small
\begin{align*}
\mathcal{Z}_I^{\Lambda}(\,\cdot\,)=\sum_{n\ge0}(-\ii
\la)^n\mathop\int\limits_{{\infi<t_1<\ldots<t_n<\supi}}\dd t_1\cdots\dd
t_n\Tr_{\res}\big[\mathcal{U}_{[t_n,\supi]}^{\Lambda}\ad(H_{\mathrm{SR}}(t_n))
\cdots\ad(H_{\mathrm{SR}}(t_1))\mathcal{U}_{[\infi,t_1]}^{\Lambda}
(\cdot)\otimes\rho_{\referres}^{\Lambda}\big]\,,
\end{align*}\normalsize
where $\infi$, $\supi$ denote the infimum and supremum of the interval $I$, respectively.
The trace over the reservoir Hilbert space can be evaluated using Wick's
theorem; see~\eqref{eq: gaussian property2}. For this purpose, we introduce the shorthand
notations
\begin{equation*}
\lone_{x,\vs}\deq (\lone_x)_{\vs}\,, \qquad \Psi_{x,\vs}(t) \deq (-\ii \Psi_{x}(t)
)_{\vs}\,,
\end{equation*}
for  $x\in\La$, $\vs\in\{ \links,\rechts \}$ (the left- and right multiplications, $(\,\cdot\,)_{\vs}$, were introduced in~\eqref{eq:2.5}). We denote by
$\mathrm{Pair}(n)$, the set of pairings of $2n$ elements and denotes by $\underline{x}$, $\underline{\vs}$ elements of $\Lambda^{2n}$, $\{\links,\rechts\}^{2n}$, respectively. In this notation, the formal Dyson series for $\mathcal{Z}_I$ can, upon
using Wick's theorem, be written as
\begin{align}\label{expansion neu 1}
 \mathcal{Z}^{{\Lambda}}_I=\sum_{n\ge0}\mathop{\int}\limits_{\infi<
t_1<t_2<\ldots t_{2n}<\supi}&\left(\prod_{i=1}^{2n} \dd
t_i\right)\sum_{\underline{x},\underline{\vs}}\sum_{\pi\in\mathrm{Pair}(n)}\zeta^{
{\Lambda}}((\underline{x},\underline{t},\underline{\vs}),\pi)\nonumber\\
&\times\caU^{{\Lambda}}_{[t_{2n},\supi]}\lone_{x_{2n},\vs_{2n}}\caU^{{\Lambda}}_{[
t_{2n-1},t_{2n}]}\cdots\lone_{x_1,\vs_1}\caU^{{\Lambda}}_{[\infi,t_1]}\,,
\end{align}
where $\zeta^{\Lambda}$ denotes a reservoir correlation function given by
\begin{align}\label{definition of zeta}
 \zeta^{{\Lambda}}((\underline{x},\underline{t},\underline{\vs}),\pi)\deq\prod_{(r
,s)\in\pi}\lambda^2   h^{{\Lambda}}(t_s, t_r,\vs_s, \vs_r)
\delta_{x_r,x_s}\,,
\end{align}
with
\begin{align}\label{eq:6.200bis}
 h^{{\Lambda}}(u,v ,\vs, \vs') \deq  \begin{cases} -\hat\psi^{{\Lambda}}(u-v)\,,&\textrm{if }
\vs=\links\,,\phantom{\rechts}\vs'=\links\,, \\
-{\hat\psi^{{\Lambda}}(v-u)}\,,&\textrm{if }\vs=\rechts\,,\phantom{\links}\vs'=\rechts\,,\\
\hat\psi^{{\Lambda}}(v-u)\,,&\textrm{if }\vs=\rechts\,,\phantom{\links}\vs'=\links\,,\\
{\hat\psi^{{\Lambda}}(u-v)}\,,&\textrm{if }\vs=\links\,,\phantom{\rechts} \vs'=\rechts\,.
\end{cases}
\end{align}
The reservoir two-point correlation function $\hat\psi^{{\Lambda}}$ has been
defined in~\eqref{eq: clear presentation finite volume correlationfunction one}.

In the next step, we decompose the expansion~\eqref{expansion neu 1} into a
sum/integral over irreducible pairings:  A pairing~$\pi$ is irreducible whenever for any $m =1, \ldots, 2n-1$, there is a pair $(r,s) \in \pi$ such that $s \leq m < r$.  

To that end, we define an operator $\caV_I^{\Lambda}$, by
\begin{align}
 \caV^{{\Lambda}}_I&\deq\sum_{n=0}^{\infty}\mathop{\int}\limits_{
\infi=t_1<\ldots<t_{2n}=\supi}\left(\prod_{i=2}^{2n-1} \dd
t_i\right)\sum_{\underline{x},\underline{l}}\sum_{\substack{\pi\in\mathrm{Pair}
(n)\\ \pi\textrm{ is
irreducible}}}\zeta^{{\Lambda}}((\underline{x},\underline{t},\underline{\vs}),
\pi)\nonumber\\
&\qquad\qquad\times\lone_{x_{2n},\vs_{2n}}\caU^{{\Lambda}}_{[t_{2n-1},t_{2n}]}\,
\cdots\caU^{{\Lambda}}_{[t_1,t_2]}\,\lone_{x_1,\vs_1}\,,
\end{align}
where the second sum is over \textit{irreducible} pairings, and we only
integrate over $2n-2$ time coordinates, with $t_1=\infi$ and
$t_{2n}=\supi$ fixed. It is then easy to check that expression~\eqref{expansion neu
1} can be rewritten as
\begin{align}\label{expansion neu 2}
\mathcal{Z}^{{\Lambda}}_{I}=\sum_{l\ge0}\mathop{\int}\limits_{
\infi<t_1<\ldots<t_l<\supi}\dd
\underline
t\,\,\mathcal
U^{{\Lambda}}_{[t_{2l},{\supi}]}\mathcal{V}^{{\Lambda}}_{[t_{2l-1},t_{2l}]}\,
\mathcal
U^{{\Lambda}}_{[t_{2l-2},t_{2l-1}]}\cdots\mathcal{V}^{{\Lambda}}_{[t_1,t_2]}\,
\mathcal U^{{\Lambda}}_{[\infi,t_1]}\,.
\end{align}
Representation~\eqref{expansion neu 2} is the starting point for the proof of
Lemma~\ref{lem: thermodynamic dyn}. The details of our proof can be found in~\cite{paper2}; Section~5.1. Here we just indicate some of the main ideas underlying it: First, we recall the bound on $\caU_I$
from~\eqref{eq: propagation bound}:
\begin{align}\label{thomas combes 2}
 |\caU^{\Lambda}_I(x,y,x',y')|\le C\e{ct
\|\im\epsilon\|_{\infty,\nu}}\e{-\nu|x-x'|-\nu|y-y'|}\,,
\end{align}
which holds uniformly in $\Lambda\subseteq\mathbb{Z}^{d}$. We observe that Assumption~\ref{ass:
analytic dispersion} ensures that $\tdl\caU_I^{\Lambda}=\caU_I$, in the sense of convergence of
kernels. Furthermore, the reservoir two-point correlation functions
$\Psi^{\Lambda}(t)$ and $\Psi(t)$ are bounded uniformly in $t$ and
$\Lambda$, and, by Assumption~\ref{ass: exponential decay}, we have that
$\tdl\Psi^{\Lambda}(t)=\Psi(t)$, uniformly in $t$ on compact subsets of $\mathbb{R}$. One then proves that
$\caV_{I}^{\Lambda}$ defines a bounded operator on $\scrB_2(\scrH_{\sys})$, for any
$\Lambda$, including $\Lambda=\Z^d$,  and that $\tdl\caV_{I}^{\Lambda}=\caV_I$,
in the sense of kernels. It is then straightforward to show that the expansion
for $\caZ_I^{\Lambda}$ converges absolutely in norm as an operator on
$\scrB_2(\scrH_{\sys})$, the bounds being uniform in $\Lambda$, and that $\tdl\caZ_I^{\Lambda}$ has a
limit $\caZ_I$, in the sense of convergence of kernels. Moreover, 
repeated use of~\eqref{thomas combes 2} reveals that
\begin{align}
 |\caZ_I^{\Lambda}(x,y,x',y')|\le
C\e{-\nu|x-x'|-\nu|y-y'|}\,,\quad\quad|\caZ_I(x,y,x',y')|\le
C\e{-\nu|x-x'|-\nu|y-y'|}\,,
\end{align}
for some $\nu>0$, where the constant $C$ is independent of $\Lambda\subseteq\mathbb{Z}^{d}$
 and is uniform in $I$, for $I$ contained in compact subsets of $\mathbb{R}$. Thus, for any
exponentially localized density matrix $\rho_{{\sys}}\in\scrB_1(\scrH_{\sys})$ and any
finite time $t$, we can define
\begin{align}
 \langle
O(t)\rangle_{\rho_{\sys}\otimes\rho_{\referres}}=\Tr_{\sys}[O\caZ_{[0,t]}\rho_{\sys}]
\deq\tdl\Tr_{\sys}[O^{\Lambda}\caZ_{[0,t]}^{\Lambda}\rho_{\sys}^{\Lambda}]\,,
\end{align}
where $O^{\Lambda}=\lone_{\Lambda}O\lone_{\Lambda}$, with $O\in{\mathop{\mathfrak{A}}\limits^{\circ}}$ or $\frX$.
By the same reasoning, one also establishes that the infinite-volume objects are translation-invariant, i.e., $\caA=\caT_{-y}\caA\caT_y $, for
$\caA=\caZ_{[0,t]}, \caU_I,\caV_I$ and $y\in\Z^d$. This becomes plausible if one recalls that translation-invariance in finite volume was broken only because of the Dirichlet boundary condition.

The advantage of representation~\eqref{expansion neu 2}
over~\eqref{expansion neu 1} is that, after Laplace
transformation, it can be resummed: For $z\in\C$, with $\re z$ sufficiently large, we define
\begin{align}\label{11eq}
 \caR(z)\deq\int_0^{\infty}\dd t\,\e{-zt}\caZ_{[0,t]}\,.
\end{align}
In order to identify the leading contributions to $\caZ_{[0,t]}$ and $\caR$,
respectively, we define an operator
\begin{align}\label{definition of the dressing operator}
 \caV_I^{(2)}&\deq\sum_{n=2}^{\infty}\mathop{\int}\limits_{\infi=t_1<\ldots<t_{
2n}=\supi}\left(\prod_{i=2}^{2n-1} \dd
t_i\right)\sum_{\underline{x},\underline{l}}\sum_{\substack{\pi\in\mathrm{Pair}
(n)\\ \pi\textrm{ is
irreducible}}}\zeta^{{}}((\underline{x},\underline{t},\underline{\vs}),
\pi)\nonumber\\
&\qquad\qquad\times\lone_{x_{2n},\vs_{2n}}\caU^{{}}_{[t_{2n-1},t_{2n}]}
\cdots\caU^{{}}_{[t_1,t_2]}\lone_{x_1,\vs_1}\,,
\end{align}
and the Laplace transforms
\begin{align}\label{laplace transfroms neu}
 \caM(z)\deq\int_0^{\infty}\dd t\,\left(\caV_{[0,t]}-\caV_{[0,t]}^{(2)}
\right)\,,\quad\quad \caR_{\mathrm{ex}}(z)\deq\int_0^{\infty}\dd t\,\e{-zt}\caV_{[0,t]}^{(2)}\,.
\end{align}
Recalling the definition of $\zeta$ in~\eqref{definition of zeta} and of $\caV$ and $\caV^{(2)}$, we observe that, roughly speaking, $\caM(z)$ contains all contributions to second order in $\lambda$ from the correlation functions $\zeta$, but higher orders of $\lambda$ enter $\caM(z)$ trough the field term $\lambda^2\field\cdot X$.

Recall the definition of the operator $\caJ_{\theta}$ in~\eqref{eq:def jkappa}.
\begin{lemma}\label{lem: pseudoresolvent neu}
The operator-valued function  $(z,\theta) \mapsto \caJ_{\theta} \caA(z)
\caJ_{-\theta}$, with $\caA= \caM, \caR_{\mathrm{ex}}$ is
analytic  in the region  $|\theta|<k_\theta, \mathrm{Re}\,z>-k_z$, for some
$k_{z},k_\theta >0$,  and satisfies the bounds (as $\la\to 0$)
\begin{equation}  \label{eq: bounds on carex neu}
\sup_{|\theta|<\theta_0\,,\,\mathrm{Re}\,z>-g_0}  \left\{  \begin{array}{rr} 
\|\mathcal{J}_{\theta}\mathcal{M}(z)\mathcal{J}_{-\theta}\|
&= \,\,\caO(\lambda^2) \,,   \\[2mm]
\|\mathcal{J}_{\theta}\mathcal{R}_{\mathrm{ex}}(z)\mathcal{J}_{-\theta}\|
&=\,\,\caO(\lambda^4)  \,.  \\[2mm]
  \end{array} \right.
\end{equation}
Moreover, for $\re z>0$, 
\begin{align}\label{eq:6.4 neu}
\mathcal{R}(z) &=(z-\caL_\sys-\mathcal{M}(z)-\mathcal{R}_{\mathrm{ex}}(z))^{-1}\,,
\end{align} 
where $\caL_\sys=\ad (H_\sys)$ is the Liouvillian of the particle system.
\end{lemma}
The proof of this Lemma is contained in~\cite{paper2}, Section~5.3. Here we just sketch the main key ideas. First, one observes that the time integrals in~\eqref{expansion neu 2} are convolutions. Thus, when taking the Laplace transform, it suffices to consider the Laplace transforms of $\caU_{[0,t]}$ and $\caV_{[0,t]}$. The former being given by the resolvent of $\caL_{\sys}=\ad(H_{\sys})$, it suffices to consider the operators $\caV_{[0,t]}-\caV_{[0,t]}^{(2)}$ and $\caV_{[0,t]}^{(2)}$, respectively. The Laplace transform of $	\caV_{[0,t]}-\caV_{[0,t]}^{(2)}$ can be computed explicitly (see below), and the claims concerning $\caM$ in Lemma~\ref{lem: pseudoresolvent neu} can be checked easily. It remains to analyze $\caV_{[0,t]}^{(2)}$, as defined in~\eqref{definition of the dressing operator}. The sum over the spatial coordinates $\underline{x}$ in~\eqref{definition of the dressing operator} can be bounded using the Combes-Thomas bound for the propagators $\caU_{I}$ (see~\eqref{eq: propagation bound}), at 
the 
price of a `mild' 
exponential growth in time; (it is `mild' because  $\nu$ on the right hand side of~\eqref{eq: propagation bound} can be chosen arbitrarily small, as long as it does not depend on $\la$). To bound the integrals over the time coordinates $\underline{t}$, we use the exponential decay in time of the correlation function $\zeta$. We can  cope with the `mild' exponential growth coming from the sum over the spatial coordinates by slightly reducing the decay rate in the exponential decay coming from the correlations $\zeta$. Then we are left with the problem of analyzing a one-dimensional `gas' of pairings between points confined to an interval of the real line, with integrable (in fact, exponential) decay in the distance between points in each pair.  It has been remarked repeatedly that, for such systems, one can integrate over all times and sum over all possible pairings. For more details, we refer to an earlier paper~\cite{deroeckfrohlichpizzo} and to the companion paper~\cite{paper2}. 
\\

From here on, our analysis proceeds as follows: In the next subsection, we explicitly
calculate the term $\caM(z)$ and show that, in some sense to be made precise, it
is close to the generator $M$ of the linear Boltzmann equation discussed in
Section~\ref{sec: kinetic}. In a next step, carried out in Section~\ref{sec:
analysis of resolvent around zero}, we show that $\caR(z)$ is comparable to
the resolvent of $M$ when restricted to fibers corresponding to `small' momenta.

\subsection{Calculation of $\mathcal{M}(z)$}\label{sec: calculation ladders}
We start with the calculation of $\caM$ defined in~\eqref{laplace transfroms neu}, i.e.,
\begin{align} \label{eq: expression cal z}
 \mathcal{M}(z)=&\la^2 \int_0^{\infty} \dd t\, \e{-zt}
\sum_{x\in\Z^d}\,\sum_{\vs_1,\vs_2\in\{ \links,\rechts \}}   h(0,t, \vs_1,\vs_2)
\lone_{x, \vs_2} \e{-\ii t \ad(H_\sys)}\lone_{x, \vs_1}\,.
 \end{align}
We recall the fiber decomposition introduced in Section~\ref{sec: fiber
decomposition}. Identifying the fiber spaces $\scrH_p$ with
$\mathrm{L}^2(\tor)$, we interpret $\caM(z)_p$ as an operator acting on
$\mathrm{L}^2(\tor)$.
 The action of the unitary group $\e{-\ii t H_{\sys}}$, with 
\mbox{$H_{\sys}=T-\lambda^2\field\cdot X$}, is given by
 \begin{align}
(\e{-\ii tH_{\sys}}f)(k)=f(k-\lambda^2{\field}t)\e{-\ii \Phi_k(t)}\,,\qquad\
f\in\mathrm{L}^2(\tor)\,, \label{eq: explicit unitary}
\end{align}
where $\Phi_k(t)\deq \int_0^t\dd
s\,\varepsilon(k-\lambda^2{\field}s)$. 
We split $\mathcal{M}(z)$ into $\mathcal{M}(z)= \sum_{\vs_1,\vs_2}\caM^{\vs_1,\vs_2}(z)$,
corresponding to the second sum in \eqref{eq: expression cal z}.
A tedious but straightforward calculation, using~\eqref{eq: explicit unitary}, 
yields the fiber operators
 \begin{align*}
 (\mathcal{M}^{\links\links}(z)_pf )(k)&=-\lambda^2\int_0^{\infty}\dd
t\,\hat\psi(t)\,\int_{\tor}\dd k'\e{-zt-\ii
\Phi_{k'+\frac{p}{2}}(t)+\ii\Phi_{k-\frac{p}{2}}(t)}
f(k -\lambda^2t \field)\,,\nonumber\\[1mm]
 (\mathcal{M}^{\rechts\rechts}(z)_p f)(k)&=-\lambda^2\int_0^{\infty}\dd
t\,\hat\psi(-t)\,\int_{\tor}\dd k'\e{-zt-\ii
\Phi_{k+\frac{p}{2}}(t)+\ii\Phi_{k'-\frac{p}{2}}(t)}f (k -\lambda^2t
\field)\,,\nonumber\\[1mm]
 (\mathcal{M}^{\links\rechts}(z)_pf)(k)&=\lambda^2\int_{0}^{\infty}\dd
t\,\hat\psi(-t)\,\int_{\tor}\dd k'\e{-zt-\ii
\Phi_{k'+\frac{p}{2}}(t)+\ii\Phi_{k-\frac{p}{2}}(t)}f(k' -\lambda^2t
\field)\,,\nonumber\\[1mm]
( \mathcal{M}^{\rechts\links}(z)_pf) (k)&=\lambda^2\int_{0}^{\infty}\dd
t\,\hat\psi(t)\,\int_{\tor}\dd k'\e{-zt-\ii
\Phi_{k+\frac{p}{2}}(t)+\ii\Phi_{k'-\frac{p}{2}}(t)}
f(k'-\lambda^2t\field)\,,\nonumber
 \end{align*}
 where $f\in\mathrm{L}^2(\tor)$.
 
\subsection{Analysis of ladder diagrams}\label{sec:
analysis ladders}
The idea of our analysis is to expand the restrictions to the fibers $\scrH_{\lambda^2\kappa}$ of the operators 
$\mathcal{R}(z)$, i.e.,
$(\caR(z))_{\lambda^2\kappa}$, for small $z$, around the contributions of order $\la^2$.
Note that, because of the small fiber momentum, $\lambda^{2}\kappa$, there are no contributions of order $1$. To capture these contributions, we will define an operator 
$\widetilde M\equiv \widetilde M^{\la,\ka,\field}$ acting on
$\mathrm{L}^2(\tor)\simeq\scrH_{\lambda^2\kappa}$ that satisfies
 \begin{align}  \label{eq: deltam rules}
(\caL_\sys+\mathcal{M}(0) + \caR_{\mathrm{ex}}(0))_{\lambda^2\kappa} =
\lambda^2\widetilde M + \caO(\lambda^4(1+\str \ka \str)) \,,
\end{align}
as $\lambda\to 0$, $\kappa\to 0$. Note that we have set the spectral
parameter~$z$ in $\caM$ and $\caR_{\mathrm{ex}}$ to zero. As the notation
suggests, $\widetilde M$ is closely related to $M\equiv M^{\ka,\field}$, the
operator introduced and analyzed in Section~\ref{sec: kinetic}; (note that, for
$\field=0$, $\widetilde M=M$). The norms below refer to
$\scrB(\mathrm{L}^2(\tor))$.
\begin{lemma}\label{lem: trivia of delta m}
Define  $\widetilde M \deq M+\delta M$ with
\begin{equation*}
\delta M=\delta M^{\field,\la} \deq K(\field)-K(0), \qquad \textrm{where}\,\,  
K(\field) = (\caM(z=0,\field))_{\la^2\ka}\,.
\end{equation*}
The  operator $\delta M$ is bounded, $\norm \delta M^{\field,\la}\norm \leq C$,
and relatively bounded w.r.t.\ to  $\field\cdot\nabla$ and $M$, with a bound of
order~$\la^2$. More precisely,
\begin{equation}  \label{eq: relative bound}
\norm \delta M f \norm \leq \la^2 C ( \norm  f \norm+  \norm \field\cdot\nabla f
\norm)  \,,
\end{equation}
for any function $f$ in the domain of $\field\cdot\nabla$. Furthermore,
\eqref{eq: deltam rules} holds.
\end{lemma}
\begin{proof}
By inspection of the expressions for $\caM$ in Section~\ref{sec: calculation
ladders}, we get
\begin{equation*}
\norm K(\field)f - K(0)  f \norm \leq  C\int \d t \,\str\psi(t) \str   \left(
\min(\str \la^2 t \field \str, 1)\norm f \norm + \norm f(\cdot+ \la^2 t
\field)-f(\cdot) \norm \right)\,,
\end{equation*}
and \eqref{eq: relative bound} follows from the exponential decay of $\psi(t)$.  
To prove \eqref{eq: deltam rules} we verify that
\begin{align*}
\la^{-2}(\caL_\sys)_{\la^2 \ka} & =\ii   \ka\cdot \nabla \varepsilon -  \field
\cdot \nabla +  \caO(\la^2 \ka)\,,  \\
\la^{-2}(\caM(z=0, \field=0))_{\la^2 \ka}  &=  G+L  +\caO(\la^2 \ka)\,,
\end{align*}
by explicit computation. (The $\links$-$\links$ and $\rechts$-$\rechts$ terms in
$\caM$ give rise to $L$, the mixed ones to $G$.)  

\end{proof}

We conclude with the remark that, presumably, Lemma \ref{lem: trivia of delta m} cannot be improved to 
\beq \label{eq: far desire}
\lim_{\field \to 0}\norm \delta M  \norm=0\,.
\eeq 
Such an estimate can easily be obtained for the $\links$-$\rechts$ and $\rechts$-$\links$  terms but not for $\links$-$\links$ and $\rechts$-$\rechts$. 
To get a feeling for this, let us consider a constant dispersion law, $\varepsilon(k)=\varepsilon(0)$  (which would actually violate our assumptions, but this should not matter here). Then $\Phi_k(t)=t\varepsilon(0)$, and, by spectral calculus,
\beq
(\mathcal{M}^{\links\links}(0)_p+\mathcal{M}^{\rechts\rechts}(0)_p)f (k)=- (2 \pi)^d \lambda^2\int_{-\infty}^{\infty}\dd
t\,\hat\psi(t)\
f(k -\lambda^2t \field)  =  - (2 \pi)^d \lambda^2  \psi(\ii \field \cdot \nabla) f (k)\,.
\eeq
Obviously $\lim_{\field \to 0}\norm \psi(\ii \field \cdot \nabla) - \psi(0) \norm \to 0$ holds only if the function $\psi$ is constant.

\subsection{Analysis of $\widetilde M$} \label{sec: analysis tilde m}
In this subsection, we show that, in a small open neighborhood of
the origin, $\widetilde M$ has an isolated simple eigenvalue $\la^2$-close to
that of $M$. We recall the definition of the gap $g\kin(\field)$ (see Section~\ref{sec:
kinetic}) and define~$B_r$ to be the disk $B_r \deq\left\{ z \in \bbC\,:\,\str z \str \leq r
\right\}$. 

\begin{lemma}\label{lem: spectrum of tilde m}
 There is a constant $r>0$, $r \propto g\kin(0)$, such that, inside the ball $B_r$, $M$ and $\widetilde M=M+\delta M$  have unique simple
eigenvalues $\eig\kin\equiv\eig\kin(\ka,\field)$ and $\eig\kini\equiv
\eig\kini(\la,\ka,\field) $, respectively, with $\str
\eig\kini-\eig_{M^{\phantom{'}}}\str =\caO(\la^2)$.  Moreover, for $z \in B_r$, 
\begin{equation*}
\frac{1}{z-\widetilde M } = \frac{1}{z-u\kini} P\kini + \caO(z^0)\,.
\end{equation*}
\end{lemma}
\begin{proof}
For $M$, this has already been proven at $\ka=0$ in Section~\ref{sec: kinetic} and
extended to $\ka \neq 0$ by using perturbation theory of isolated eigenvalues
(Section~\ref{sec: asymptotics of semigroup}). 
Since $\delta M$ is relatively bounded w.r.t.\ $M$, we can again use
perturbation theory to prove the claim for $\widetilde M$.
Estimating the resolvent $(z-\widetilde M)^{-1}$ by using a Neumann series expansion in $\delta
M$ and applying \eqref{eq: relative bound}, we obtain that
\begin{align} \label{eq: relative resolvent bound}
\norm \delta M (z-M)^{-1} \norm &  \leq   \la^2 C\left( 1+ \str z\str +  \norm
(z-M)^{-1} \norm \right) \,,
\end{align}
which can be used to complete the proof of the lemma. (We refer the reader to \cite{katoperturbation}
for details on the perturbation theory for isolated eigenvalues.)

\end{proof}

Although this will not be used in our analysis, it is worthwhile pointing out an important
difference between the spectral analysis of $M$ and that of $\widetilde M$:
Thanks to the resolvent bound~\eqref{lem: resolventbound}, we know that the spectrum of
$M$, apart from the eigenvalue $\eig$, is bounded away from the real axis. In the
nomenclature of Section~\ref{sec: asymptotics of semigroup}, $g\kin(\ka,\field) >0$.
In fact, by Lemma~\ref{lem: kinetic}, item $iv$, we have an explicit bound, uniform in
$\field$, on the real part of eigenvalues with large imaginary part. For $\widetilde M$,
analogous statements do \textit{not} hold, because  $\delta M$ is not small in
norm, but only relative to $M$; see \eqref{eq: relative resolvent bound}.
 To guarantee that $z \notin \si(\widetilde M)$, the right-hand side of~\eqref{eq:
relative resolvent bound} should be strictly smaller than one. Clearly, for any
fixed $\la$, this is not the case when $\im z \to \pm\infty$.

\section{Analysis of $\caR(z)$ around $z=0$} \label{sec: analysis of resolvent
around zero}
In this section, we analytically continue the operator $\caR(z)$ (see \eqref{11eq}), a priori only defined for $\re
z>0$, to the region $\{z\in\C\,:\, |z|< \lambda^2 r\}$, for some
$r>0$ and $\la$ sufficiently small. This is accomplished by applying
perturbation theory to the fiber operators  $(\caR(z))_{\lambda^{2}\kappa}$. The guiding
idea is that  $(\caR(z))_{\la^2 \ka}$ is a small perturbation of $(z-\lambda^2
M^\ka)^{-1}$, where $M^{\ka}$ has been analyzed in Section~\ref{sec: kinetic}.
However, it turns out to be more convenient to replace $M$ by the operator $\widetilde M$
introduced in Section~\ref{sec:
analysis ladders}. In Section~\ref{sec: perturbation kinetic limit}, we implement the perturbation
theory developed on the basis of  Lemma~\ref{lem: pseudoresolvent neu}.  The
small parameters are the coupling constant $\la$, the (scaled) fiber
momentum $\ka$ and the field $\field$. All these three parameters must be assumed
to be sufficiently small throughout our analysis, and we do not repeat this assumption in every step.

\subsection{Perturbation around the kinetic limit}\label{sec: perturbation
kinetic limit}
 Recall the definition of
$\mathcal{R}(z)$ in equation~\eqref{11eq}.
\begin{align*}
\mathcal{R}(z)=\int_0^{\infty}\dd t\,\e{-zt}\,\Matz\,.
\end{align*}
 The main results of this subsection state that the operator 
 $(\caR(z))_{\la^2\ka}$ has a unique simple pole in a neighborhood of $z=0$, 
 whose residue, $P\equiv P^{\lambda,\ka,\field}$, is a rank-one operator (see
Lemma~\ref{lemma: the pole is simple}) with the property that, in the fiber
indexed by $\kappa=0$,
\begin{align}
 P^{\la,\ka=0,\field}=\ket\zeta\bra1\,,\qquad\textrm{ with
}\qquad\|\zeta-\zeta_M\|_{\mathrm{L}^2(\tor)}=\caO(\lambda^2)\,,
\end{align}
where $\zeta_M$ is the invariant state of the generator, $M$, of the linear
Boltzmann evolution; see Section~\ref{sec: kinetic}. This result is stated in
Lemma~\ref{cor:laurent}. Moreover, we show that $P^{\la,\ka,\field}$ is an
analytic function of $\kappa$ and a regular function of $\field$; see
Lemma~\ref{lem: asym perturbation}. On an intuitive level, this means that the
long-time dynamics of $(\caZ_{[0,t]})_{\la^2\ka}$, is dominated by the linear
Boltzmann evolution $\e{tM}$. This statement is formalized in Section~\ref{sec:
equilibrium case}, for $\field=0$, and in Section~\ref{sec: proof of main results}.2,
for~$\field\not=0$.

To start with, we define an operator, $S$, acting on $\mathrm{L}^2(\tor)$ by
\begin{align}\label{eq: definition of S}
 S \equiv &   S(z,\field,\la, \ka) \deq  (\caL_\sys+ \caM(z) +
\caR_{\mathrm{ex}}(z))_{\la^2 \ka}\,.
\end{align}
Note that $(\mathcal{R}(z))_{\la^2 \ka}= (z-S)^{-1}$ (whenever both sides are well-defined),  
and that $S(z)$ is a closed operator on~$\mathrm{L}^2(\tor)$.
It  is bounded except for the term $ \field\cdot \nabla$  that originates from
$\caL_\sys$. 

For simplicity, we often abbreviate $S(z,\field,\la,\ka)$ by, for example, $S(z)$, when we
consider the operator-valued function \mbox{$z\mapsto S(z)$}, with the other variables kept
fixed. We use similar shorthand notation for $u_M\equiv
u_M(\lambda,\kappa,\field)$, $P\equiv P^{\lambda,\kappa,\field}$, etc.\ in this
and the remaining sections.

Let $\caD \subset \mathrm{L}^2(\bbT^d)$ be the dense subspace of real-analytic
functions on $\bbT^d$; cf.\ Section~\ref{sec: kinetic}.  Recall the constant
$k_z$ from Lemma~\ref{lem: pseudoresolvent neu}. 

\begin{lemma}\label{lem: prop for asymptotic}
$ $
\begin{itemize}
 \item[$i.$] $\caD$ is a core for $S$ and $S \caD \subset \caD$. For all $ z 
\in \bbC$ satisfying $\re z \geq -k_z $ and such that $ (z-S(z))^{-1}$ exists (i.e., as a bounded operator), we have that $(z-S(z))^{-1} \caD \subset
\caD$. 
 \item[$ii.$]  The differences $S(z)-S(z=0)$ and $S(\ka)-S(\ka=0)$ are bounded
operators, and they are analytic in the variables $\ka,z$
 in the region
\mbox{$\mathrm{Re}z>-k_z$} and $|\kappa|<k_\theta$.

\end{itemize}
\end{lemma}
\begin{proof}
$\caD$ is a core for $S$ because it is a core for $\field\cdot \nabla$. Further,
we first establish that, for $\ga \in \bbC^d$ sufficiently small,
  \begin{equation} \label{eq: stability of sz}
   \e{\ga\cdot \nabla} S\, \e{-\ga \cdot\nabla}-S = \caO(\ga)\,. 
\end{equation}
Note that for $\theta\in\C^{2d}$ sufficiently small,
$ \caJ_\theta \caA   \caJ_{-\theta}-\caA=\caO(\theta)$, where $\caA=\caL_\sys
+\caM(z) + \caR_{\mathrm{ex}}(z) $. This follows from the analyticity of the
dispersion law $\ve$ in $\caL_\sys$ and from Lemma~\ref{lem: pseudoresolvent neu} for
$\caM(z) + \caR_{\mathrm{ex}}(z)$. The bound \eqref{eq: stability of sz} is then
obtained by restricting to a fiber. 
A Neumann series expansion of $(z-S(z))^{-1}$, using \eqref{eq: stability of sz}, for
some sufficiently small $\ga$ (depending on $z$), yields boundedness of 
$    \e{\ga \cdot\nabla} (z-S(z))^{-1}  \e{-\ga \cdot\nabla}$. Together with
\eqref{eq: stability of sz}, this implies part $i$, after an application of the
Paley-Wiener theorem.

To prove part $ii$, it suffices to observe that the term $-\field \cdot \nabla$ in $S$ is
independent of $z $ and $\ka$.

\end{proof}
Next, we argue that the condition $z\in\sigma(S(z))$ has a unique solution $z\adj$ in a
neighborhood of $z=0$:
\begin{lemma}
Fix some $r>0$ sufficiently small, e.g., $r= g\kin(0)/4$. Then there is a unique
$z=z^{*}(\la,\ka,\field)$ in $B_{\la^2 r}$ such that $z-S(z)$ is not invertible,
i.e., such that $z \in \sigma(S(z))$. This unique $z^*$ is an isolated simple
eigenvalue of $S(z^*)$.
\end{lemma}
\begin{proof}
We write
\begin{equation*}
S(z)= \la^2 \widetilde M   +  \la^2 (\caM(z)-\caM(0))_{\la^2 \ka}  + 
(\caR_{\mathrm{ex}}(z))_{\la^2 \ka} =: \la^2\widetilde M+  A(z,\la)\,.
\end{equation*}
Recall that $  \widetilde M $ has a unique, simple eigenvalue $u_{\kini}$ in
$B_{r}$, for some $r>0$.   Since \[\norm A(z,\la) \norm \leq C (\la^4+\la^2 \str
z\str) \leq C \la^4\,,\] for  $z \in B_{\la^2r}$, an application of spectral
perturbation theory shows that $S(z)$ has a unique simple eigenvalue~$s(z)$ in (to be concrete)
 the disk $B_{3\la^2 r/4}$.  This eigenvalue is given by
\begin{equation} \label{eq: explicit eigenvalue}
s(z)  =  \la^2 u_{\kini}+  \Tr \left[  P_{\kini} A(z,\la) \right]+ \caO(\norm
A(z,\la) \norm^2)\,,
\end{equation}
where $P\kini$ is the spectral projection of $\widetilde M$ associated with the
eigenvalue $u\kini$ and the trace is over~$\mathrm{L}^2(\tor)$.

Next, we show that there is a unique $z\in B_r$ such that $z\in\sigma(S(z))$. First, we show uniqueness of $z$: Assume that there are two solutions $z_1,z_2$ of  $z \in \si(S(z))$. Then 
 \begin{equation*}
\str  z_1-z_2 \str = \str s(z_1)-s(z_2) \str  \leq   C  \str z_1-z_2 \str
\sup_{z \in B_{\la^2 r }}\left\| \frac{\partial}{\partial z}A(z,\la)\right\|
\leq  C\la^2 \str z_1-z_2 \str\,,
 \end{equation*}
which is a contradiction, for $\la$ small enough. From \eqref{eq: explicit
eigenvalue}, we also get  $\str s(z) - \la^2 u_{\kini} \str \leq C \str \la
\str^4$. 

Second, we show that there is at least one $z\in B_r$ such that $z\in\sigma(S(z)) $: Assume there is no solution of $z \in \si(S(z))$. 
By taking $\ka$ sufficiently small, we can make sure that $ u_{\kini}$ lies in
the ball $B_{r/4}$.
 Then, denoting by $\caC$ the positively oriented integration contour 
\begin{equation}\label{eq: the integation contour}\caC\deq\{ z\,:\, \str z \str
= \la^2r/2\}\,,
\end{equation} we note that
\begin{equation*}
\sup_{z \in \caC}\left\norm   \frac{1}{z- \la^2\widetilde M} \right\norm \leq 
C\la^{-2}, \qquad    \sup_{z \in \caC}\left \norm   \frac{1}{z- S(z)}
\right\norm \leq C\la^{-2}\,.
\end{equation*}
The first bound is stated in Lemma~\ref{lem: spectrum of tilde m}. The second
bound follows from the first one by Neumann series expansion, using the bound on
$A(z,\la)$.   Next, we note the identity
\begin{equation*}
 \frac{1}{z-S(z)}- \frac{1}{z- \la^2\widetilde M}=     \frac{1}{z-S(z)} A(z,\la)
\frac{1}{z- \la^2\widetilde M}\,
\end{equation*}
and integrate both sides along the contour $\caC$. The right-hand side is
bounded in norm by $C$, thus, after integration, it is bounded by $C \la^2$. On the left-hand
side, the contour integral of $\frac{1}{z- \la^2\widetilde M}$ yields the
spectral projection~$P_{\kini} $. We therefore arrive at a contradiction with the
assumption that $ \frac{1}{z-S(z)}$ has no singular points. 
\end{proof}

\begin{lemma}\label{lemma: the pole is simple}
The pole at $z^*$ of $z \mapsto (z-S(z))^{-1}$  is simple and its residue, $P$, is
a rank-one operator.
\end{lemma}
\begin{proof}
Simplicity of the pole has already been established in the above proof.
To get hold of the residue, we expand $S(z)$ in some neighborhood of $z^*$: 
\begin{equation*}
z-S(z) = \underbrace{z-z^*}_{=:y} -  \underbrace{(S(z^*)-z^*)}_{=: F_0} -
\underbrace{(z-z^*)S_1(z^*)}_{=:yF_1} - \underbrace{S_{>1}(z-z^*)
}_{=:F_2(y)}\,,
\end{equation*}
where $S_1(z^*)= \partial_z S(z^*)$ and $S_{>1}(z-z^*)= \sum_{n \geq 2} \frac{(z-z^*)^{n}}{n!} (\partial_z)^n S(z^*)$. First, we rewrite
\begin{equation} \label{eq: linearized f}
\frac{1}{y-F_0-yF_1} =         \frac{1}{y- (1-F_1)^{-1}F_0} \cdot\frac{1}{1-F_1}  \,.
\end{equation}
From the considerations above, we know that $F_0$ has an isolated simple
eigenvalue at~$0$.  It follows that~$0$ is also an eigenvalue of
$(1-F_1)^{-1}F_0$.   Since $F_1 = \caO(\la^2)$, $((1-F_1)^{-1}-1)F_0$  is a relatively bounded perturbation of~$F_0$ with small relative bound and hence perturbation theory
ensures that this eigenvalue is again simple and isolated, and we call $P_F$ the
corresponding one-dimensional spectral projector.  It follows that $ P_F
\frac{1}{1-F_1} $  is the residue of the function~\eqref{eq: linearized f} at
$y=0$. Then, we write
\begin{equation} \label{eq: split of quadratic}
 \frac{1}{y-F_0-yF_1-F_2(y)}  =   \frac{1}{1 +
(y-F_0-yF_1)^{-1} F_2(y)} \cdot \frac{1}{y-F_0-yF_1}  \,.
\end{equation}
Since $\norm F_2(y)\norm \leq C\str y \str^2$ and $\norm (y-F_0-yF_1)^{-1} \norm
\leq  C\str y \str^{-1}$, the first factor on the right side of~\eqref{eq: split of quadratic} is
analytic in a neighborhood of $y=0$ and it reduces to $1$ at $y=0$. It follows
that the residue at $y=0$ of the function~\eqref{eq: split of quadratic} is $
P_F  \frac{1}{1-F_1} $. Since $P_F$ is one-dimensional, this is a rank-one
operator. 
\end{proof}
To continue, we denote the  pole by
$u(\la,\ka,\field)=z^{*}(\la,\ka,\field)$.

In the statement of  the next lemma, it is convenient to extend $S$ to a
function of $\ga$ defined in a complex neighborhood of $\gamma=0$ by setting $S(\ga)\deq
\e{\ga \cdot \nabla} S \e{-\ga \cdot \nabla}$. As argued in the proof of
Lemma~\ref{lem: prop for asymptotic}, $\ga \mapsto S(\ga)$ is analytic, and the
statements of the Lemmas~\ref{lem: prop for asymptotic} and~\ref{lemma: the pole
is simple} remain valid for sufficiently small $\ga$. 
We write
\begin{equation}\label{eq: splitting of R}
 (\caR(z))_{\la^2 \ka} = \frac{1}{z-S(z)}=  \frac{1}{z-u} P^{\ka} +
R^{\ka}(z)\,,
\end{equation}
with $z \mapsto R^{\ka}(z)$ bounded and analytic in $B_{\la^2r}$, for some $r>0$.
In the following, we often use the shorthand notations $P\equiv P^{\ka}$ and $R(z)\equiv
R^{\ka}(z)$.

\begin{lemma} \label{lem: asym perturbation}
The pole $u$ and the operators $P^{\ka},R^{\ka}(z)$ are analytic in $\ka$ and
$\ga$.
\end{lemma}

\begin{proof}
Residue and pole can be expressed as contour integrals, namely
\begin{equation} \label{eq: contour rep}
P =  \frac{1}{2 \pi \ii}    \int_\caC \d z \,\frac{1}{z-S(z)}, \qquad  \eig   P=
\frac{1}{2 \pi \ii}     \int_\caC \d z  \,\frac{z}{z-S(z)}\,,
\end{equation}
where the contour $\caC$ is defined in~\eqref{eq: the integation contour}. Then the
analyticity in $\ka$ and $\ga$ follows from the analyticity of $S(z)$, as
established in Lemma~\ref{lem: prop for asymptotic}.  
 \end{proof}

We now summarize our findings and derive some additional algebraic properties of
the residue $P$.
\begin{lemma}\label{cor:laurent} For $\kappa=0$, the residue, $P$, can be written as $P^{\kappa=0} = \str \zeta^{\lambda,\kappa=0,\field} \rangle \langle 1
\str$, with $\zeta\equiv\zeta^{\la,\ka, \field} \in \mathrm{L}^2(\tor)$ a
real-analytic function satisfying
$$  \norm \zeta -\zeta_M \norm =\caO(\lambda^2)\,,$$
where $\zeta_M$ is the invariant state of $M$; see Section~\ref{sec: kinetic}. 
For $\ka=0$, $\zeta$ is a probability density on $\tor$.
The function $\eig\equiv\eig(\la,\ka,\field) \in \bbC $ satisfies $$\overline{
\eig(\ka)}=\eig(-\overline\ka)\,, \qquad |\eig-\la^2 \eig\kin |=\caO(\lambda^2)\,. $$
Moreover, we have that
\begin{align}\label{eq: eigenvalue is zero}
 \eig(\kappa=0)=0\,,\qquad P^{\ka=0} R^{\ka=0}(z)=0\,.
\end{align}

\end{lemma}
\begin{proof}
For an arbitrary exponentially localized density matrix $\rho_\sys$,
\begin{equation*}
\overline{\Tr_{\sys}[\e{{\ii}\lambda^2\ka\cdot X}\caZ_{[0,t]}\rho_\sys]} =
\Tr_{\sys}[ (\e{{\ii}\lambda^2\ka \cdot X})^*  (\caZ_{[0,t]}\rho_\sys)^* ] =
\Tr_{\sys}[ \e{-{\ii}\lambda^2\overline\ka \cdot X} \caZ_{[0,t]}\rho_\sys ]\,,
\end{equation*}
where we have used that $\caZ_{[0,t]}$ preserves positivity, in particular,
hermiticity. Writing this in terms of fibers, taking the Laplace transforms
and comparing the singular parts, we get  $\eig(-\bar\ka)=\overline{
\eig(\ka)}$. 
For real~$z$, the operator $\caR(z) = \int \d t\,  \e{-t z} \caZ_t $ preserves
positivity, hence $ \caR(z) \rho_\sys  \geq 0$, for an arbitrary positive-definite
$\rho_\sys$, which implies the positivity of the function   $P^0 (\rho_\sys)_0$,
where $0$ refers to the zero fiber. Writing $P^\ka= \str \zeta^\ka
\rangle \langle \tilde \zeta^\ka \str $, it then follows that $\zeta^0$ can be
chosen to be positive.   
Since $\caZ_{[0,t]}$ preserves the trace, it follows that $\Tr_\sys
\left[\caR(z)\rho_\sys\right]= 1/z$. Hence, from~\eqref{eq: splitting of R}
(evaluated in the fiber indexed by $\ka=0$),
\begin{equation*}
\frac{1}{z} =  \frac{1}{z-u(\ka=0)}  \langle 1, P^{0} (\rho_\sys)_0
\rangle+\langle 1, R^0(z)\,(\rho_{\sys})_0\rangle =    \frac{1}{z}  \langle 1,
\zeta^0 \rangle \langle \tilde \zeta^0,  (\rho_\sys)_0 \rangle+\langle 1,
R^0(z)\,(\rho_{\sys})_0\rangle \,.
\end{equation*}
Since this identity has to hold for any density matrix $\rho_\sys$, and because $R^0(z)$ is
an analytic function, we conclude that $u(\ka=0)=0$, $P^0R^0(z)=0$ and $\tilde
\zeta^0=c 1$. Choosing the normalization $c=1$, it follows that 
$\int_{\tor} \zeta^0=1$, i.e.,~$\zeta^0$ is a probability density. The
analyticity of $\ga \mapsto \e{\ga \cdot \nabla} P \e{-\ga \cdot \nabla} $
implies boundedness of $\e{\ga \cdot \nabla}\zeta^{\ka}, \e{\ga \cdot
\nabla}\tilde\zeta^{\ka}$, and hence analyticity of $k \mapsto \zeta^{\ka}(k),
\tilde\zeta^{\ka}(k)$. The bounds by terms $\caO(\la^2)$ follow immediately from the
perturbation theory outlined above.
\end{proof}
\vspace{2cm}

\section{Proof of main results}\label{sec: proof of main results}
Before we prove our main results for $\field\not=0$, we first address the equilibrium regime, $\field=0$.

\subsection{The equilibrium regime}\label{sec: equilibrium case} 
If the field $\field$ vanishes, our results can be strengthened, because $\delta
M=0$.  For $\field=0$, the function $z \mapsto (\caR(z))_{\la^2 \ka}$ has only one
pole, $\eig(\la,\ka,\field=0)$, in the region  $ \re z > -\la^2g\kin(\ka,
\field=0) +\caO(\la^4)$; (cf.~the remark following Lemma~\ref{lem: spectrum of
tilde m}).  Then the pole $\eig(\la,\ka,\field=0)$ determines the long-time
properties of the dynamics.  By applying an inverse Laplace transform, one then easily proves the following theorem.
\begin{theorem}\emph{[Equilibrium asymptotics]} \label{thm: equilibrium rte}
We set $\field =0$. 
Then, for $0<\la$ and $\ka$ sufficiently small, there is  a constant $g >0$ such that
\begin{align*}
\left\|(\mathcal{Z}_{[0,t]})_{\lambda^2\kappa}-\e{ \eig(\ka) t}P^\ka\right\|=\caO\big(\e{-
g\la^2  t} \big)\,, \qquad \textrm{as } t \to \infty\,,
\end{align*}
as operators on $\mathrm{L}^2(\tor)$.
\end{theorem}
Note that $g$ can be chosen as $g= g\kin(0)/5$, for example. Also recall that
$u(\kappa=0)=0$, by Lemma~\ref{cor:laurent}. For the proof, we refer to
Theorem~4.5 of~\cite{deroeckfrohlichpizzo}.

\subsection{Proof of Theorems~\ref{thm: stationary} and~\ref{thm: diffusion}}\label{section8.2}
\begin{proof}[Proof of Theorem~\ref{thm: stationary}]
We first set $\field=0$. Let $f$ be a continuous function on $\tor$, (hence {
$M_f\in{\mathop{\mathfrak{A}}_{\mathrm{ti}}}$)}. Then Lemma~\ref{lemma: fibers} and Theorem~\ref{thm: equilibrium rte} yield 
\begin{align*}
 \Tr_{\sys}[ M_f \caZ_{[0,t]}{\rho_\sys}] =  \langle f, \zeta^{0}
\rangle+\caO(\e{-g \lambda^2 t})\,,\quad\quad\textrm{as}\quad t\to\infty\,,
\end{align*}
proving~\eqref{results approach to ness 2}.

To prove~\eqref{results approach to ness 1}, we choose $\field\not=0$ and define $X\deq\scrB(\mathrm{L}^2(\tor))$.  Then the function $t \mapsto
\caZ_{[0,t]}$ is in 
$\mathrm{L}^{\infty}(\bbR_+ ,X)$. A standard Tauberian theorem, see
e.g.~\cite{tauber} or~\cite{hillephilips}, states the equivalence of  
\begin{align*}
\lim_{T\to\infty}\frac{1}{T}\int_{0}^{T}\,\dd t\,(\Matz)_0=x\,, \qquad
\textrm{and} \qquad
\lim_{z\to 0,\,\re z\geq 0}z\int_0^{\infty}\,\dd t\,\e{-zt}({\Matz})_0=x\,,
\end{align*}
for some $x\in X$.  Existence of the first limit yields Theorem~\ref{thm:
stationary}.  We show that the second limit
exists: By Equation~\eqref{eq: splitting of R} we have that, for
 $|z|$ sufficiently small,
\begin{align} \label{eq:general ztimesres}
z\int_0^{\infty}\,\dd
t\,\e{-zt}({\Matz})_0=\frac{z}{z-\eig(\kappa=0)}P^{0}+zR^0(z)\,.
\end{align}
Since $\eig(0)=0$ and $z \mapsto R^{0}(z)$ is analytic in a neighborhood of
$z=0$, the limit equals $P^{0}$. 
\end{proof}

\begin{proof}[Proof of Theorem~\ref{thm: diffusion}]
We start from the identity \[\langle X(t) \rangle_{\rho_\sys\otimes\rho_\referres}= \int_0^t
\d s \,\langle V(s) \rangle_{\rho_\sys\otimes\rho_\referres}+\langle X
\rangle_{\rho_\sys\otimes\rho_\referres}\,, \]
which follows from the definition of the velocity operator~\eqref{def: finite volume velocity finite} in
 finite-volume and can be easily justified, using Lemma~\ref{lem: thermodynamic dyn}, in the thermodynamic limit, for exponentially localized density matrices $\rho_\sys$; see Section 5.2 of~\cite{paper2} for details. Hence
\begin{align} \label{eq: ergodic x}
\frac{1}{t} \langle X(t)\rangle_{\rho_\sys \otimes
\rho_\referres}&=\frac{1}{t}\int_0^{t}\dd s\, \langle V(s)
\rangle_{\rho_\sys \otimes \rho_\referres}+ \frac{1}{t} \langle
X(0)\rangle_{\rho_\sys \otimes \rho_\referres} \nonumber\\[1mm]
&  =\langle\nabla\varepsilon,\zeta^{\field,\la}\rangle+   \caO({t}^{-1})\,,
\end{align}
as $t\to\infty$. Since $\rho_\sys$ is exponentially localized,  the second term
on the right-hand side of the last line vanishes as $t\to \infty$, and we conclude that
$v(\field)=\langle\nabla\varepsilon,\zeta^{\field,\la}\rangle$. The statement that
$v(\field)\not=0$, for $0<|\field|$ sufficiently small, follows from
$v_M(\field)\not=0$, for $\lambda$ small enough.

We define the diffusion tensor $D\equiv D(\field)$ by
\begin{align}\label{def:diffusionconstant}
D^{ij}\deq-
\la^{-4}\frac{\partial^2}{\partial\kappa^i\kappa^j}\bigg|_{\kappa=0}
\eig(\kappa)\,,
\end{align}
where the factor $\la^{-4}$ is attributed to the fact that the fiber momentum is $\la^2
\ka$, rather than $\ka$.
From  \mbox{$\overline{\eig(\kappa)}=\eig(-\overline{\kappa}) $}, we conclude
that $D$ has real entries. For $\field\not=0$ sufficiently small, positive-definiteness follows from the fact that~$u(k)$ is a $C^{\infty}$ function in $\field$ (see Lemma~6.1 in~\cite{paper2}), and the positive-definiteness of $D(\field=0)$. It remains to argue that the above definition in terms of the
eigenvalue $u$ is equivalent to the one given in~\eqref{eq:diffconst}, namely, 
\begin{align*}
D^{ij}= \lim_{T\to \infty}\frac{1}{T^2}\int_{0}^{\infty}\,\dd
t\,\mathrm{e}^{-\frac{t}{T}}\,\langle
(X^{i}(t)-v^i(\field)t)(X^{j}(t)-v^j(\field)t)\rangle_{\rho_\sys \otimes
\rho_\referres}\,.
\end{align*}
This is straightforward, and we omit details. But we do explain an analogous argument
relating the expression for the asymptotic velocity $v(\field)$ in terms of the eigenvalue to the
one involving moments of $X$; i.e., we check that
\begin{align} \label{eq: residu v}
v(\field)= \langle \nabla\varepsilon,\zeta^{\field,\la}\rangle=
\ii\frac{\partial}{\partial\kappa}\bigg|_{\kappa=0} \eig(\kappa)\,.
\end{align}
It follows from our discussion in Section~\ref{sec: fiber decomposition} that
\begin{align*}\Tr_{\sys}[ X \rho_{\sys,t}] =\ii
\frac{\partial}{\partial\kappa}\bigg|_{\kappa=0}\langle
1,(\rho_{\sys,t})_{\lambda^2\kappa}\rangle\,,
\end{align*}
and we obtain, using $\Tr_{\sys}[ X \rho_{\sys,t}]=\caO(t)$, that
\begin{align*}
z^2\int_{0}^{\infty}\,\dd t\,\e{-zt}\Tr_{\sys}\left[ X \rho_{\sys,t}\right]&=z^2\ii
\int_{0}^{\infty}\,\dd
t\,\e{-zt}\frac{\partial}{\partial\kappa}\bigg|_{\kappa=0}\langle
1,(\rho_{\sys,t})_{\lambda^2\kappa}\rangle\nonumber\\
&=z^2\ii \frac{\partial}{\partial\kappa}\bigg|_{\kappa=0}\int_{0}^{\infty}\,\dd
t\,\e{-zt}\langle 1,(\rho_{\sys,t})_{\lambda^2\kappa}\rangle\,,\nonumber
\end{align*}
for $\mathrm{Re}\,z>0$. By straightforward manipulations, using \eqref{eq:
ergodic x}, the limit $z \to0$ of the left-hand side equals $v(\field)$.  We
abbreviate $\frac{\partial}{\partial\kappa}\big|_{\ka=0}f(\ka)$ by $f'(0)$. 
Thanks to Lemma~\ref{cor:laurent}, it follows that
\begin{align*}
z^2\int_{0}^{\infty}\,\dd t\,\e{-zt}\Tr_{\sys}\left[ X \rho_{\sys,t}\right]&=\ii z^2
\langle 1,\frac{1}{z-\eig(0)}P^{0}(\rho_{\sys})_0   \rangle' + \ii z^2\langle
1,R^0(z)\left(\rho_{\sys})\right)_0 \rangle' \\[1mm] 
&= \frac{ \ii z^2 {\eig}'(0)}{(z-\eig(0))^2}\langle1,P^{0}(\rho_{\sys})_0
\rangle+
 \frac{\ii z^2}{z-\eig(0)}\langle1,P^{0}(\rho_{\sys})_0
\rangle'+\caO(z^2)\nonumber\\[2mm] 
&=\ii {\eig}'(0)+\caO(z)\,,
\end{align*}
as $z\to 0$, where we have used that $P^{\kappa}, R^{\ka}(z)$ are analytic in
$z$, $\kappa$ and that $\eig(0)=0$, $\langle1,P^{0}(\rho_{\sys})_0 \rangle=1$. 
Passing to the limit $z\to 0$, we confirm \eqref{eq: residu v}. 
\end{proof}

Note that, with our present methods, we cannot prove the existence of the
limit
\begin{align*}
\lim_{T\to \infty}\frac{1}{T}\int_{0}^{T}\,\dd t\,\frac{1}{t}\langle
(X^{i}(t)-v^{i}(\field)t)(X^{j}(t)-v^{j}(\field)t)\rangle_{\rho_\sys \otimes
\rho_\referres}.
\end{align*}
The problem is that Tauberian theorems only hold if one can bound the integrand
by a constant, whereas we only have a rough a priori bound, namely $\langle X^2(t)
\rangle_{\rho_{\sys}\otimes\rho_\referres} \leq C \str t\str^2$. When $\chi\neq 0$ we therefore have to
state our results in terms of Laplace transforms of time-dependent quantities in $t$.

\newpage
\bibliographystyle{plain}
\bibliography{biblio}

\end{document}